%% file: main.tex
\documentclass{article}

\usepackage{arxiv}

\usepackage{amsmath,amsfonts,amssymb, amsthm}

\usepackage[usenames,dvipsnames]{xcolor}
\definecolor{niceRed}{RGB}{190,38,38}
\definecolor{niceBlue}{HTML}{0466a7}

\usepackage[utf8]{inputenc} 
\usepackage[T1]{fontenc}    
\usepackage{hyperref}       
\hypersetup{
	colorlinks = true,
	urlcolor = {black},
	linkcolor = {niceBlue},
	citecolor = {niceRed} 
}

\usepackage{natbib}
\bibpunct[, ]{(}{)}{,}{a}{}{,}%
\def\BIBand{and}%

\usepackage{xurl}           
\usepackage{booktabs}       
\usepackage{multirow}
\usepackage{amsfonts}       
\usepackage{nicefrac}       
\usepackage{microtype}      
\usepackage{cleveref}       
\usepackage{graphicx}
\usepackage{doi}

\usepackage{subcaption} 
\usepackage{booktabs}	
\usepackage{multirow}
\usepackage{float}

\usepackage{mathtools}

\DeclareMathOperator*{\argmax}{arg\,max}
\DeclareMathOperator*{\argmin}{arg\,min}

\usepackage[ruled]{algorithm2e} 

\newtheorem{theorem}{Theorem}[section]
\newtheorem{proposition}{Proposition}[section]

\DeclareMathOperator{\Acal}{\mathcal{A}}
\DeclareMathOperator{\Gcal}{\mathcal{G}}
\DeclareMathOperator{\Ical}{\mathcal{I}}
\DeclareMathOperator{\Lcal}{\mathcal{L}}

\DeclareMathOperator{\Pcal}{\mathcal{P}}
\DeclareMathOperator{\Scal}{\mathcal{S}}
\DeclareMathOperator{\Vcal}{\mathcal{V}}
\DeclareMathOperator{\Xcal}{\mathcal{X}}
\DeclareMathOperator{\R}{\mathbb{R}}
\DeclareMathOperator{\E}{\mathbb{E}}
\newtheorem{result}{Result}

\title{Revenue in First- and Second-Price Display Advertising Auctions: Understanding Markets with Learning Agents}

\newif\ifuniqueAffiliation
\uniqueAffiliationtrue

\ifuniqueAffiliation 
\author{
	Martin Bichler \\
	\small Department of Computer Science\\
	\small Technical University of Munich\\
	\small \texttt{m.bichler@tum.de} \\
	\And
	 Alok Gupta\\
	\small Carlson School of Management\\
	\small University of Minnesota \\
	\small \texttt{alok@umn.edu} \\
	\And
	Matthias Oberlechner\\
	\small Department of Computer Science\\
	\small Technical University of Munich \\
	\small \texttt{matthias.oberlechner@tum.de}\\
}
\else
\usepackage{authblk}

\newbox{\orcid}\sbox{\orcid}{\includegraphics[scale=0.06]{orcid.pdf}} 
\author[1]{%
	\href{https://orcid.org/0000-0000-0000-0000}{\usebox{\orcid}\hspace{1mm}David S.~Hippocampus\thanks{\texttt{hippo@cs.cranberry-lemon.edu}}}%
}
\author[1,2]{%
	\href{https://orcid.org/0000-0000-0000-0000}{\usebox{\orcid}\hspace{1mm}Elias D.~Striatum\thanks{\texttt{stariate@ee.mount-sheikh.edu}}}%
}
\affil[1]{Department of Computer Science, Cranberry-Lemon University, Pittsburgh, PA 15213}
\affil[2]{Department of Electrical Engineering, Mount-Sheikh University, Santa Narimana, Levand}
\fi

\renewcommand{\headeright}{}
\renewcommand{\undertitle}{}
\renewcommand{\shorttitle}{}

\hypersetup{
	pdftitle={Low Revenue in Display Ad Auctions: Algorithmic Collusion vs. Non-Quasilinear Preferences},
	pdfsubject={cs.GT},
	pdfauthor={Martin~Bichler, 	Alok~Gupta, Laura~Mathews, Matthias~Oberlechner},
	pdfkeywords={algorithmic collusion, display ad auctions, equilibrium learning} ,
}

\begin{document}
\maketitle

\begin{abstract}
	The transition of display ad exchanges from second-price to first-price auctions has raised questions about its impact on revenue. 
	Auction theory predicts the revenue equivalence between these two auction formats. 
	However, display ad auctions are different from standard models in auction theory in at least two important ways. First, automated bidding agents cannot easily derive equilibrium strategies in first-price auctions 
	because distributional information regarding competitors' values or even the number of competitors is not readily available. Second, due to principal-agent problems, bidding agents typically maximize return-on-investment and not payoff.
	The literature on learning agents for real-time bidding is growing because of the practical relevance of this area; most research has found that due to the lack of complete information to the agents, learning agents do not converge to an equilibrium. Specifically, research on algorithmic collusion in display ad auctions has argued that first-price auctions can induce symmetric Q-learning agents to tacitly collude, resulting in bid values below Nash equilibrium, which in turn leads to lower revenue compared to the second-price auction. Whether bids are in equilibrium or not cannot easily be determined from field data since the underlying values of bidders are unknown. In this paper, similar to the literature on algorithmic collusion, we draw on analytical modeling and numerical experiments and explore the convergence behavior of widespread online learning algorithms in both complete and incomplete information models. Contrary to prior results, we show that there are no systematic deviations from equilibrium behavior. 
	We also explore the differences in revenue from first- and second-price auctions. Prior research has not compared the revenue implications of the two auction formats for utility functions relevant in this domain such as return-on-investment. 
	We show that learning algorithms also converge to equilibrium, but revenue equivalence does not hold. 
	Our results show that collusion may not be the explanation for lower revenue with first-price auctions.  Instead, we show that even with equilibrium behaviors, second-price auction achieves higher expected revenue compared to the first-price auction, i.e., the change in auction format might have had substantial and non-obvious consequences for ad exchanges and advertisers bidding for display ads.
\end{abstract}



\section{Introduction}\label{sec:intro}

Real-time bidding is a means by which advertising inventory is sold on a per-impression basis \citep{choi2020online}. Display ad auctions are a prime application where advertisers compete for ads. Digital advertising is an industry with \$129 billion spent in the U.S. in 2019, which is more than half of total media advertising spending.\footnote{\url{https://skai.io/blog/monday-morning-metrics-digital-ad-spending-more-than-half-us-media-spend/}, accessed March 18, 2024}
In 2022, more than 90\% of all digital display ad spend were transacted via real-time bidding \citep{yuen22}. Moreover, the real-time bidding market size is expected to grow substantially in the next years \citep{tunuguntla2021near}. The Google Ad Manager, OpenX, PubMatic, and Xandr are among the largest ad exchanges operating these auctions. Not surprisingly, there is a growing literature in Information Systems and the Management Sciences on display ad auctions \citep{mehta2020sustaining, sayedi2018real, choi2020online, christopher2022bypassing, balocco2024lemon}. 

In the past, most ad exchanges used the second-price auction. Consequently, most of the literature on display ad auctions focuses on repeated second-price auctions under different models \citep{conitzer2022multiplicative, chen2023complexity, ciocan2021tractable}. 
However, since 2019, several of them have begun to experiment with first-price auctions and are being used by major advertising exchanges \citep{despotakis2021first}. Google claimed this switch was to ``help advertisers by simplifying how they buy online ads.''\footnote{\url{https://support.google.com/adsense/answer/10858748?hl=en}, accessed March 18, 2024} One complaint by advertisers regarding second-price auctions was regarding the price floors that sellers introduced to enhance yield. These floors were motivated by the observation that the second-price auctions' revenue was often significantly below the first-price auctions. However, advertisers considered price floors as opaque and there were concerns about manipulation on part of the sellers.\footnote{\url{https://www.adexchanger.com/data-driven-thinking/are-artificial-price-floors-the-next-iteration-of-ad-fraud/}, accessed August 18, 2024} Such concerns do not arise with first-price auctions. 

While first-price auctions were adopted by large exchanges, others still offer second-price auctions, and the choice of the auction format is a matter of ongoing debate. 
Some studies have observed decreased bid prices after the switch \citep{alcobendas2021adjustment}, and some ad exchanges reported significant revenue losses after the switch to first-price auctions.\footnote{\url{https://www.adexchanger.com/programmatic/unraveling-the-mystery-of-pubmatics-5-million-loss-}\\ \url{from-a-first-price-auction-switch/}, accessed October 1, 2024} All of them report that the adaptation of bidding strategies was not instant but took time for agents to adapt. Obviously, if only one bidder starts to shade his bid below and the others don't he would lose much more often. 
Overall, the empirical analysis of this policy change is challenging due to many confounding factors \citep{goke2021bidders}. Changes in the economy lead to changes in advertising spending over time. Ad exchanges switched at different times, which might have shifted some of the demand from one to another exchange. Several factors such as the type of impression, the day of the week, the season, and the type of publisher, all affect prices. Due to changes in macro and microeconomic conditions, several confounding factors and limited availability of data (such as advertiser valuations), it is difficult for ad exchanges to identify the causal impact of a change in their auction format on their revenue.\footnote{This statement is based on private communication with a large ad exchange.} 



The value of economic models lies in their ability to perform comparative statics, allowing for the analysis of how isolated changes in exogenous parameters affect economic outcomes while holding other factors constant (ceteris paribus). Even though the details of real-world display advertising markets such as the exact valuations or the bidding strategies implemented, might differ, we aim to understand how auction format changes impact revenue in a model that captures key characteristics of display ad auctions. 

Standard game-theoretical models of first- and second-price auctions are Bayesian games in which agents are assumed to follow Bayes-Nash equilibrium (BNE) strategies that are determined ex-ante. The celebrated revenue equivalence theorem predicts that the first- and the second-price auctions achieve the same revenue in expectation in a model with independent prior valuations \citep{vickrey1961counterspeculation, myerson1981OptimalAuctionDesign}. However, display ad auctions are different in at least two important ways: the bidding strategies are learned and the objective of the bidding agents is return-on-investment (ROI) rather than payoff (aka. a quasilinear utility function).

\subsection{Learning agents and algorithmic collusion}

In display ad auctions bidding is fully automated, and bids need to be submitted in milliseconds. In a first-price auction, bidders want to shade their bid below their valuation for an impression to make a profit. However, they cannot simply follow a pre-determined equilibrium bidding strategy. Without knowing the valuations of others or even only a prior distribution of valuations, deriving an equilibrium bidding strategy is impossible. Automated bidding agents need to learn a good strategy while bidding. Not surprisingly, after the move to first-price auctions, a growing literature using (multi-armed) bandit algorithms in display ad auctions has emerged \citep{liang2024online, kumar2024strategically, wang2023learning, zhang2023online,ai2022no, zhang2022leveraging, han2020optimal, tilli2021multi}. 
Bandit algorithms are fast and adapt to changes in the environment quickly. So, their use for automated bidding is natural. There is extensive literature on such online algorithms, and they are widely used for online optimization problems \citep{lattimore2020bandit}. 
Vanishing regret (aka. no regret) is the central design goal, but there are also popular bandit algorithms without vanishing regret.

However, bidding in repeated auctions is not a standard online learning problem. Advertisers compete against each other in a game-theoretical environment and their actions influence each others' profit. It is not clear that algorithms designed for online problems would perform well in game-theoretical environments. 
Actually, it is well-known that online learning algorithms in games may cycle, end up in off-equilibrium states, or even in chaotic dynamics \citep{sanders2018prevalence}. The analysis of learning dynamics in games can be arbitrarily complex in general \citep{andrade2021learning}, and not much is known about online learning algorithms in auction games.

Among the off-equilibrium outcomes of learning dynamics, \textit{algorithmic collusion} captured the most attention in the broader public \citep{Calvano.2020, Klein.2021, hansen2021frontiers, abada2023artificial}. Algorithmic collusion in pricing games such as a Bertrand oligopoly refers to algorithms that converge to a supra-competitive price higher than the Nash equilibrium. The topic got the attention of regulators \citep{OECD.2017} as it might jeopardize consumer welfare when algorithms are used for pricing on digital platforms. The jury about algorithmic collusion in pricing is still out, and there are arguments for but also against this being a concern in practice in pricing \citep{Boer.2022}. 
Display ad auctions are highly automated markets and algorithmic collusion might arise here as well. In such markets, we have millions of repeated interactions among bidders for similar types of impressions. For example, the Google Display Network serves 24.17 billion impressions per day.\footnote{\url{https://www.business2community.com/online-marketing/how-many-ads-does-google-serve-in-a-day-0322253}}
Indeed, \citet{banchio2022artificial} reported the results of numerical experiments with bidding agents based on specific Q-learning algorithms in a complete-information model, where the value of bidders is public knowledge and the competition doesn't change. 
Interestingly, they find that first-price auctions with no additional feedback lead to \textit{tacit-collusive outcomes}, while second-price auctions do not. Low revenue in first-price auctions is an obvious concern to ad exchanges, who might lose billions of dollars. However, this is also important to understand for advertisers, who might either decide to use such algorithms or exploit the off-equilibrium behavior of others in order to increase their payoff. Therefore, a key question in this paper is:  \textit{Can we assume that the repeated interaction of bidding agents in display ad auctions leads to an equilibrium?} If this is not the case, then this is a concern for ad exchanges, advertisers, and publishers alike. 

\subsection{Return-on-investment}

If learning agents that maximize payoff converge to equilibrium, then the celebrated revenue equivalence theorem \citep{vickrey1961counterspeculation} holds, and we can assume that the change from second- to first-price auctions is without loss. 
Given that bidding is highly automated, advertisers need to rely on specialized demand-side platforms (DSPs) to do the bidding on their behalf. DSPs are independent firms, and they don't just maximize the profit of the advertiser without a limit. Rather they are given a budget by the advertiser for a specific marketing campaign. The advertiser does not want any money back but wants to have the budget used most effectively. 
There is a large (academic and practitioner) literature on advertising and, in particular, on display ad auctions suggesting that return-on-investment (ROI) or sometimes return-on-spend (ROS) are the primary objectives (see Section \ref{sec:util_models}). 
Revenue equivalence does not need to hold if bidders do not maximize payoff. Unfortunately, deriving equilibrium and understanding the revenue ranking of auctions for such non-quasilinear utility functions analytically leads to intractable non-linear differential equations. For such problems, we do not even have an exact mathematical solution theory, and no equilibrium strategies are known for bidders that maximize ROI or ROS.

\subsection{Contributions}

We make two main contributions. In a \textit{first contribution}, we aim to analyze the outcome of a market with learning agents and explore the likelihood of collusion more broadly. 
While there is extensive literature suggesting that auto-bidders use bandit algorithms, different types of such algorithms are available. Through extensive numerical experiments, we explore a variety of learning algorithms that compete against the same or different algorithms. We analyze two models: The \textit{complete-information environment} mimics a stylized competition of advertisers with a predetermined value for an impression of a particular type, e.g. two luxury car manufacturers bidding for users who are interested in expensive cars. The complete-information model was used in earlier studies on algorithmic collusion in auctions for its simplicity \citep{banchio2022artificial}. 
Even in high-frequency display ad auctions, this model might be too simple. Impressions are not all the same and they differ in the attributes of the user. Therefore, even though there are millions of display ad auctions per day, competitors are not exactly the same in each auction and the advertisers participating in a specific auction are better modeled as draws from a distribution. In an incomplete-information (aka. Bayesian) game model, we only need to assume a prior distribution of bidders, not specific bidders with exact values in each auction. For the analytical derivation of equilibrium bidding strategies, this prior distribution needs to be available to the analyst. However, the learning agents in our analysis do not know this distribution a priori, but they learn a strategy in a process of exploration and exploitation.\footnote{This effectively addresses one of the main criticisms of game theory, that of bidders needing strong informational assumptions, which was famously raised in the Wilson doctrine \citep{wilson1985game}.} 

We analyze both, the complete-information and the incomplete-information model and find that algorithmic collusion is an exception and not the rule. While we can replicate the findings by \citet{banchio2022artificial}, we also show that collusion is not robust even in a market with only symmetric Q-learning agents but different hyperparameters. When combining Q-learning agents with other popular learning algorithms such as Exp3 or Thompson Sampling, we do not observe collusion. 
Convergence to equilibrium might take thousands of rounds, but we also do not observe a systematic deviation from the equilibrium in earlier rounds. One might argue that the real world is less structured and in spite of millions of auctions per day, some changes in demand and supply occur. However, if we don't find algorithmic collusion in this stable repeated environment (complete and incomplete information), it is also not likely to emerge in an environment that is more dynamic. In contrast, if we observed collusion in a static environment in a complete-information model as in \citet{banchio2022artificial}, it is not clear that it will persist in more dynamic environments. 

In a \textit{second contribution}, we derive equilibrium strategies for bidders that maximize ROI and ROS in single-item auctions. Prior work showed equilibrium existence or derive Price-of-Anarchy bounds on the total value obtained by the advertisers \citep{aggarwal2019autobidding, chen2023complexity}. The problem of computing an equilibrium is computationally hard, and closed-form solutions are typically unknown for for various model assumptions \citep{aggarwal2024auto}. We take a different approach and use equilibrium learning to compute equilibria in markets with ROI and ROS bidders. Such numerical techniques to find equilibrium bidding strategies are new, but they allow us to find equilibrium even markets with such non-quasilinear utility models \citep{bichler2023soda}. For specific prior distributions such as uniform priors, we also provide analytical results.
We find that with low competition the second-price auction leads to higher expected revenue for the auctioneer than the first-price auction. As in the standard model with payoff-maximizing utility functions, the differences grow smaller with increasing competition, but the assumption of revenue equivalence is not satisfied anymore, and moving from a second- to a first-price auction is not without loss for publishers and ad exchanges. 

\section{Related Literature} \label{sec:lit}
Auctions and, in particular, online auctions have a long history in the information systems literature \citep{pinker2003managing, greenwald2010evaluating, cason2011experimental, bichler2010ResearchCommentaryDesigning, choi2018display}, beyond what we can cover in this section. We draw on specific strands in the literature close to our key research question. First, we discuss utility models, as they have been reported for display ad auctions. Second, we introduce relevant work on collusion in auctions. Finally, we cover the literature on equilibrium learning relevant to this paper.



\subsection{Utility Models of Bidders in Display Ad Auctions} \label{sec:util_models}

There is a large literature on display ad auctions and real-time bidding \citep{despotakis2021first}. The authors assume different utility models to describe the advertisers' objectives. Originally, authors used a standard quasi-linear utility function as is usual in auction theory \citep{edelman2007internet, despotakis2021first}. However, more recent literature does not assume payoff maximization anymore. The ROI became a popular metric because of marketing budget allocation \citep{szymanski2006impact, borgs2007dynamics,jin2018real, wilkens2017gsp}. Apart from payoff maximization, we analyze agents that learn to maximize expected ROI as a measure of the ratio of profit to cost:
\begin{equation*}
	\mathbb{E}[ROI] = \mathbb{E}[\tfrac{v-p}{p}],
\end{equation*}
where $v$ denotes the valuation and $p$ the payment or cost.\footnote{Note that this ratio becomes very high if prices are close to zero. In practice, there are minimum bid prices (aka. floors) and we rule out division by zero as a boundary case for our theoretical analysis in this paper.}
One reason for ROI optimization is that advertising is usually delegated from the advertising firm (the principal) to an agent. 
In display ad auctions, a demand-side platform (DSP) gets a budget and is responsible for bidding on impressions in milliseconds \citep{yuan2014survey}. This requires a sophisticated technical infrastructure that advertisers don't have. The agent might have an incentive to spend as much as possible of the principal's money. To control the actions of the agent, the DSPs are given a budget by this advertiser \citep{bichler2018principal}. ROI measures how efficiently the investment or budget for a campaign is used by the DSP and has, therefore, become a central objective for DSPs. The advertiser does not get money back, even if the payoff of an impression is low, and the DSPs are compared by the ROI they achieve.

Apart from ROI, the practitioners' literature also discusses ROS.
\begin{equation*}
	\mathbb{E}[ROS] = \mathbb{E}[\tfrac{v}{p}]
\end{equation*}
ROS can be motivated from agents that aim to maximize total value of the impressions subject to payments being less than a budget constraint $B$ \citep{tunuguntla2021near}. In an offline auction where all $M$ impressions are auctioned off at the same time, the utility of the bidder could be described as
\[\max \sum_{m=1}^M\mathbb{E}(x_m v_m) \quad 
s.t. \sum_{m=1}^M x_m p_m \leq B,
\]
where $x_m$ is a binary variable that equals one upon winning and zero otherwise, $p_m$ is the payment, and $v_m$ the value per impression. \citet{tunuguntla2021near} only require the budget constraint for a campaign to hold in expectation and derive the Lagrangian of the resulting optimization problem. Based on the first-order condition for this Lagrangian, they derive a bidding strategy $b_m = v_m/\lambda^*$, where $b_m$ is the bid, and $\lambda^*$ defines a threshold. In a first-price auction $b_m=p_m$. \citet{tunuguntla2021near} devise an algorithm to learn an optimal $\lambda^*$. With an optimal $\lambda^*$, the advertiser wins the subset of auctions for which the value per expenditure is greatest. 
By rearranging terms, the optimal shading factor $\lambda=v_m/p_m$ is the ratio of the value for an impression and its payment, similar to ROS. This is reminiscent of widespread heuristics for the knapsack problem, which rank-order items by the ratio between value and weight of objects. $\lambda \in [0,1]$ is also described as \textit{multiplicative shading factor} in recent publications on real-time bidding \citep{balseiro2019learning, Tardos2022}. 
Looking at the practitioner literature, ROI or ROS seem to be dominant in the practice of real-time bidding.\footnote{see \url{https://www.indeed.com/career-advice/career-development/roas-vs-roi},\newline\url{https://instapage.com/blog/roas-vs-roi-which-metric-should-you-use/},\newline \url{https://hawksem.com/blog/roi-vs-roas-similarities-differences/}} 


\subsection{Collusion in Auctions and Pricing}

Collusion is a key concern in the literature on auctions, and there is much literature about which auction rules are more susceptible \citep{skrzypacz2004tacit, fabra2003tacit,blume2008modeling}. 
It is often a topic in second-price auctions, less so in first-price auctions (see \citet{krishna2009auction}). 
Here, we deal with tacit algorithmic collusion, where low off-equilibrium prices arise without explicit communication among bidders. In particular, we want to understand algorithmic collusion, i.e., collusion that arises from the repeated interaction of learning agents without them being programmed for explicit collusion. 

Algorithmic collusion was first discussed in the context of oligopoly pricing. Early treatments go back to \citet{greenwald2000shopbots}. \citet{Calvano.2020} analyzes firms who play an infinitely repeated game, pricing simultaneously in each stage and conditioning their prices on history. They find that Q-learning algorithms consistently learn to charge supra-competitive prices. These prices are sustained by collusive strategies with a finite phase of punishments followed by a gradual return to cooperation. \citet{Klein.2021} shows how Q-learning is able to learn collusive strategies when competing algorithms update their prices sequentially; as opposed to \citet{Calvano.2020}, for collusion to occur in their sequential-move setting, they do not require that algorithms can condition on own and competitor past prices. Other related papers on algorithmic collusion in oligopoly pricing were written by \citet{Hettich.2021} and \citet{Asker.2022}, for example. 

\citet{banchio2022artificial} were the first to analyze collusion in the context of real-time bidding in display ad auctions. The environment is different from the stylized oligopoly pricing models. They find that when bidders use specific Q-learning algorithms to determine their bids, the auction format and other design choices can have a first-order effect on revenues and bidder payoffs. The revenues can be significantly lower in first-price auctions than in second-price auctions. In particular, first-price auctions with no additional feedback lead to tacit-collusive outcomes, while second-price auctions do not. However, a simple practical auction design choice - revealing to bidders the winning bid after the auction - can make the first-price auctions more competitive again. Yet, the paper provides evidence that in with specific implementations of symmetric Q-learning agents the revenue in display ad auctions can be lower than that of the second-price auction.

\subsection{Equilibrium Learning}
\label{sec:learning}
The literature on algorithmic collusion and the recent work on equilibrium learning in auctions are closely related and address the same underlying question. This connection has not been explored in prior literature. Equilibrium learning tries to find equilibria, while algorithmic collusion aims to identify situations when learning agents might end up in collusive, off-equilibrium states. 
Almost all of the literature on equilibrium learning deals with finite and complete-information games \citep{fudenbergLearningEquilibrium2009}. 

Earlier approaches to finding equilibria in auctions were usually setting-specific and relied on reformulating the first-order condition of the utility function as a differential equation and then solving this equation analytically (where possible) \citep{vickrey1961counterspeculation,krishna2009auction,ausubel2019CoreselectingAuctionsIncomplete}. \citet{armantier2008ApproximationNashEquilibria} introduced a BNE-computation method based on expressing the Bayesian game as the limit of a sequence of complete-information games. They show that the sequence of NE in the restricted games converges to a BNE of the original game. While this result holds for any Bayesian game, setting-specific information is required to generate and solve the restricted games. \citet{rabinovich2013ComputingPureBayesianNash} study best-response dynamics on mixed strategies in auctions with finite action spaces. 

As described earlier, the paper by \citet{banchio2022artificial} analyzes repeated first-price auctions with Q-learning agents, who both have the same fixed value, which is common knowledge. They show that Q-learning can lead to phases with collusive strategies in the first-price auction. \citet{deng2022nash} analyze a similar model analytically and they find that mean-based algorithms converge to the Nash equilibrium in this complete-information model. {Note that the results by  \citet{deng2022nash} assume that at least two bidders have the same highest value, but convergence with a single highest-value bidder is still open.} 
Mean-based means roughly that the algorithm picks actions with low average rewards with low probability. 
This class contains most of the popular no-regret algorithms, including Multiplicative Weights Update (MWU), Follow-the-Perturbed-Leader (FTPL), {and} Exp3. This is also what we find in our complete-information experiments.  

As numerical technique to compute equilibrium for non-quasilinear utility models, we draw on SODA \citep{bichler2023soda}, a gradient-based algorithm that was shown to be very versatile and allowed for the computation of BNE in a large variety of different auction models. Whenever SODA converges to a pure strategy profile, it has to be an equilibrium. We will describe these algorithms in more detail below. This allows us to verify equilibrium in situations where no analytical solution is known, in particular in those where the objective of the bidding agents are ROI or ROS maximization. 


\section{Equilibrium Strategies}\label{sec:equi}
Let us now discuss the different utility functions for display ad auctions introduced in Section \ref{sec:util_models} and study the resulting equilibrium problem. We will introduce the equilibrium problem formally and show that for all but the standard quasi-linear utility function, we end up in non-linear differential equations. Note that, in general, there is no exact solution theory for the solution to systems of non-linear partial differential equations, which also explains why closed-form equilibrium strategies are only available for simple models such as first-price auctions with quasi-linear utility functions. 

The equilibrium problem for non-quasi-linear utility functions such as ROI and ROS maximization is much more challenging and no equilibrium predictions are known. For ROI bidders we provide an analytical solution for uniform priors. For ROS, even this restricted case is intractable analytically. 

The equilibrium learning algorithms introduced later will allow us to find numerical solutions to these equilibrium problems and compare the resulting equilibria. While numerical methods are widely used in engineering and science to solve systems of differential equations, their success in auction theory was rather limited and plagued by numerical instabilities \citep{fibich2011numerical}. Equilibrium learning provides a new and powerful approach to equilibrium analysis that has not been available so far.

\subsection{Model}

Formally, single-item auctions are modeled as incomplete-information games $ \mathcal G = (\Ical, \Vcal, \Acal, F, u)$ with continuous type and action spaces. Each bidder $ i \in \Ical = \{ 1,\dots,N\}$ observes a private value (type) $v_i \in \Vcal_i \subset \R $ drawn from some prior distribution $F$ over $\Vcal := \Vcal_1 \times \dots \times \Vcal_N$. With a slight abuse of notation, we also denote the marginal distribution over $\Vcal_i$ with $F$, since we only consider symmetric bidders with independent and identically distributed valuations. The corresponding probability density function is denoted by $f$. After observing their private types, bidders submit bids $b_i \in \Acal_i \subset \R$ and receive their payoffs as given by the (ex-post) utility function $u_i: \Acal \times \Vcal_i \rightarrow \R$ with $\Acal = \Acal_1 \times \dots \times \Acal_n$.
A pure strategy is a function $\beta_i: \Vcal_i \rightarrow \Acal_i$, mapping a bidder's value to an action. Given a strategy profile $\beta = (\beta_1, \dots, \beta_n) $, 
the expected (ex-ante) utility is defined by $ \tilde u_i(\beta_i, \beta_{-i}) := \E_{v_i \sim F}[\bar u_i(\beta_i(v_i), \beta_{-i}, v_i)]$ with the ex-interim utility 
\begin{equation}
    \bar u_i(b_i, \beta_{-i}, v_i) := \E_{v_{-i}}[u_i(b_i,\beta_{-i}(v_{-i}), v_i)].
\end{equation}
We say that a strategy profile $\beta $ is a BNE if no agent can increase their expected utility $ \tilde u_i $ by unilaterally deviating:
\begin{equation}
     \tilde u_i(\beta_i, \beta_{-i}) \geq  \tilde u_i(\beta_i', \beta_{-i}) \quad \forall \beta_i',\,\, \forall i \in \Ical.
\end{equation}

In the remaining part of this section, we analyze this model for different utility models, expressed by different ex-post utility functions $u_i$. 
Let $x: \Acal \rightarrow [0, 1]^N$ be the allocation vector, where $x_i(b) = 1$ if bidder $i$ gets the item, i.e., $b_i > \max_{j \neq i} b_j$ and $x_i(b) = 0$ else, and $p: \Acal \rightarrow \R^N $ the price vector. 
As described in the previous section, we consider payoff-maximizing or quasi-linear (QL)
\begin{equation}
    u_i^{QL}(b, v_i) = x_i(b)(v_i - p_i(b)),
\end{equation}
ROI-maximizing
\begin{equation}
    u_i^{ROI}(b, v_i) = x_i(b) \dfrac{v_i - p_i(b)}{p_i(b)},
\end{equation}
and ROS-maximizing agents
\begin{equation}
    u_i^{ROS}(b, v_i) = x_i(b)\dfrac{v_i}{p_i(b)}.
\end{equation}
Note that, depending on the payment rule, the price vector takes the values $p_i(b) = b_i$ in first-price and $p_i(b) = \max_{j \neq i} b_j $ in second-price auctions for bidder $i$ receiving the item, and $p_j(b) = 0$ for all other bidders. In the following analysis, we often use the shorter notation $p_i := p_i(b)$.

\subsection{Payoff-Maximizing Bidders}

With quasi-linear bidders (QL) that maximize payoff, the equilibrium bidding strategies in a first- and second-price auction are well known. The second-price auction has a dominant-strategy BNE of bidding truthfully ($\beta(v)=v$). For the first-price auction, we get a closed-form solution for the (non-truthful) equilibrium bidding strategy via the first-order condition of the expected utility function \citep{krishna2009auction}. Given that all opponents play according to a symmetric, increasing, and differentiable strategy $\beta$, the utility for bidder $i$ submitting bid $b$ with valuation $v$ is given by
\begin{equation}
	\bar u_i (b,\beta_{-i}, v) = G(\beta^{-1}(b))(v-b),
\end{equation}
where $G(\beta^{-1}(b))$ denotes the probability of winning with bid $b$. 
Using the first-order condition $\frac{d}{db} \bar u(b,\beta_{-i}, v) = 0$, one can derive the following ordinary differential equation (ODE)
\begin{equation}
	\dfrac{d}{dv}(G(v)\beta(v)) = v G'(v).
\end{equation}

If we assume independent uniformly distributed valuations for all $N$ agents, i.e., $G(v) = v^{N-1}$, and further suppose that $\beta(0)=0$, we get the equilibrium strategy $\beta(v) = \frac{N-1}{N} v$.
Since we introduce a reserve price $r>0$ for the ROS and ROI utility models, we will also use this reserve price here for better comparison. 
For valuations $v<r$, no bidder can make a positive profit, and consequently, they would not participate in the auction ($\beta_i(v)=0$). For higher valuations, we can again use the ODE to get a BNE with the additional initial condition $\beta(r) = r$. This leads to
\begin{equation}\label{eq:eq_fpsb_ql}
	\beta(v) =  \dfrac{N-1}{N}v + \dfrac{1}{N} \dfrac{r^N}{v^{N-1}} \quad \text{ for } v \geq r.
\end{equation}
See \citet{krishna2009auction} for the equilibrium strategy with general i.i.d. valuations. 
Note that the BNE is unique in this setting \citep{chawla2013AuctionsUniqueEquilibria}. 
Revenue equivalence of the first- and second-price auction in the independent private values model with payoff-maximizing bidders are central results in auction theory discussed in related textbooks. 

\subsection{ROI-Maximizing Bidders}
As discussed earlier, ROI maximization is widely mentioned as the objective of advertisers. 
The first observation is that the second-price auction continues to be strategyproof, as with payoff-maximizing bidders.  

\begin{theorem} \label{thm:roi_sp}
	The second-price sealed-bid single-object auction is strategyproof for bidders with an ex-post utility function of $u_i(b, v_i) = x_i(b) \tfrac{ v_i - p_i(b)}{p_i(b)}$.
\end{theorem}
\begin{proof}
Consider bidder 1 and suppose that $p_1 = \max_{j\ne 1}b_j$ is the highest competing bid. By bidding $b_1$, bidder 1 will win if $b_1 \geq p_1$ and lose if $b_1 < p_1$. Let's assume that $b_1$ equals the value of this bidder. Now, suppose that he bids an amount $b'_1 < b_1$. If $b_1 > b'_1 \geq p_1$, then he still wins, and his ROI is still $(b_1-p_1)/p_1 >0$. If $p_1 > b_1 > b'_1$, he still loses. However, if $b_1 > p_1 > b'_1$, then he loses, whereas if he had bid $b_1$, he would have received a positive ROI. Thus, bidding less than $b_1$ can never increase his utility but, in some circumstances, may actually decrease it. Similarly, it is not profitable to bid more than $b_1$. If he bids $b'_1 > b_1$, then there could be $b'_1 > p_1 > b_1$, making his utility negative. If $b'_1 \geq b_1 > p_1$, then this increase to $b'_1$ would not change his ROI.
\end{proof}

Let us next discuss the first-price auction. Analogous to the model with payoff-maximizing agents, we can try to use the first-order condition to derive equilibrium strategies for ROI-maximizing bidders. The ex-interim utility for bidder $i$ with valuation $v$ and bid $b$, given the opponents strategy $\beta$, is defined by
\begin{equation} \label{eq:foc_roi}
	\bar u_i (b, \beta_{-i}, v) = G(\beta^{-1}(b)) \dfrac{v-b}{b}.
\end{equation}
Similar to the first-price auction with quasi-linear bidders, we can derive a first-order condition assuming symmetric strategies and obtain

\begin{equation}
    \frac{d}{dv} \beta(v) = \left( v \beta(v) - \beta(v)^2 \right) \dfrac{G'(v)}{vG(v)}.
\end{equation}
This first-order, non-linear ODE is in general hard to solve analytically. If we restrict our analysis to independent, uniformly distributed valuations over $[0,1]$, i.e., $F(v) = v$ and $ G(v) = v^{N-1} $, the first-order condition \eqref{eq:foc_roi} reduces to
\begin{equation}
    \dfrac{d}{dv} \beta(v) = (N-1) \left( \dfrac{\beta(v)}{v} - \dfrac{\beta(v)^2}{v^2} \right).	
\end{equation}
For this specific form of the ODE, i.e., $\tfrac{d}{dv}\beta(v) = f \big(\tfrac{\beta(v)}{v}\big)$ with $f(x) = (N-1)(x-x^2)$, we can use substitution and derive an analytical solution of the ODE by separation of variables. 
Similar to the quasi-linear case, one can argue that bidders with a valuation $v < r$ do not participate, and that they bid exactly $ r $ for $v = r$.
The resulting initial value problem leads to the following equilibrium strategies for different numbers of bidders $N$:
\begin{equation}
    \beta(v) = 
    \begin{cases}
        \dfrac{v}{1 - \log(r) + \log(v)} &\text{ for } N = 2 \\[10pt]
        \dfrac{(N-2)v^{N-1}}{(N-1)v^{N-2} - r^{N-2}} &\text{ for } N > 2.
    \end{cases}
\end{equation}

Note that we bound the bids from below to ensure that the utility is well-defined. For the analytical solution depicted in Figure \ref{fig:analytical_QL_ROI_uniform}, we assume a reservation price of $r=0.05$. 
For more general prior distributions, such a non-linear ODE is very difficult to solve with standard analytical or numerical techniques \citep{bichler2023soda}. 
\begin{figure}[h]
    \begin{center}
        \begin{subfigure}{0.45\textwidth}
            \includegraphics[width=0.9\textwidth]{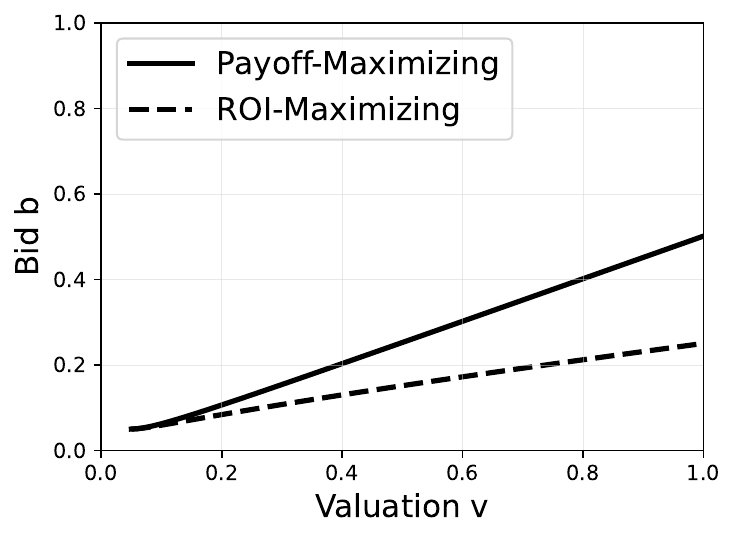}
            \caption{$N=2$ Bidders}
        \end{subfigure}
        \begin{subfigure}{0.45\textwidth}
            \includegraphics[width=0.9\textwidth]{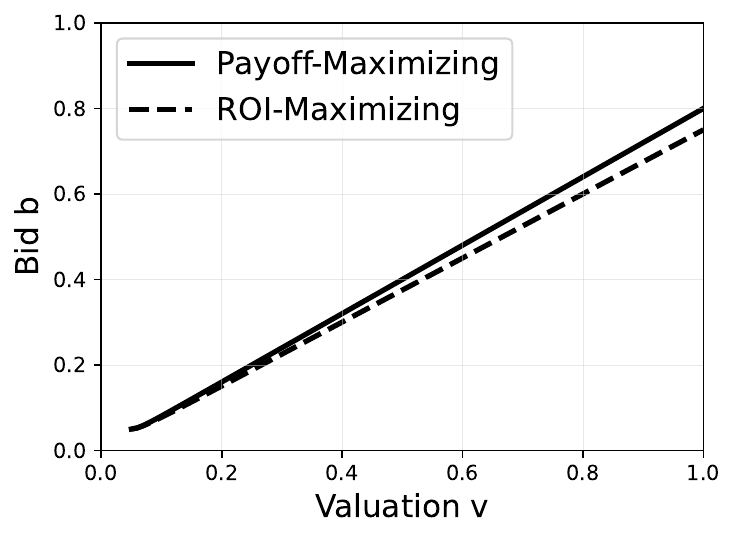}
            \caption{$N=5$ Bidders}
        \end{subfigure}
        \caption{Analytical BNE for Payoff- and ROI-Maximizing Agents in a First-Price Auction with Uniform Prior.}
        \label{fig:analytical_QL_ROI_uniform}
    \end{center}
\end{figure}

\subsection{ROS-Maximizing Bidders}
Finally, we analyze ROS-maximizing agents. As indicated earlier, there is an important difference between ROI and ROS: While the expected utility for ROI becomes negative once the price is higher than the value, ROS never becomes negative, independent of how high the payment is. In other words, with a monotone allocation rule, if a bidder is losing, they can always increase their bid without ever having a negative payoff. This leads to a situation where agents can place the highest possible bid, regardless of how high this value is and which payment rule is being used. 

\begin{proposition}
    In the first- and second-price sealed-bid single-object auction where agents are ROS maximizers, ties are broken randomly, and the action space is bounded from above, submitting the maximal bid is an equilibrium.
\end{proposition}
\begin{proof}
Assume all agents bid according to $\beta(v) = \bar b$, where $\bar b > 0$ is the maximal bid the agents can submit. Then, the ex-interim utility is $\bar u(b,\beta_{-i}, v_i) = \tfrac{1}{N} \tfrac{v_i}{\bar b} > 0$ if $b = \bar b$, and $0$ if $b < \bar b$. Since agents can only deviate from $\beta(v)$ by bidding less, this would strictly decrease their payoff. Therefore, $\beta(v) = \bar b$ is indeed an equilibrium strategy.
\end{proof}

This equilibrium strategy is also found by equilibrium learning algorithms. Obviously, this is unrealistic, because it ignores that bidders always have only a finite budget available.  
In this article, we consider ROS-maximizing bidders with a per-auction budget (ROSB), which we simply model using a log barrier function. This ROSB utility is given by
\begin{equation}
    u_i^{ROSB}(b, v_i) = x_i(b) \left( \dfrac{v_i}{p_i(b)} + \log(B - p_i(b)) \right),
\end{equation}
where $B$ is the budget. It is difficult to derive a BNE analytically for such a utility function. 
The first-order condition of the ex-interim utility $\bar u_i $ for the first-price payment rule, under the assumption of symmetric bidders, leads to
\begin{align}
    \label{eq:foc_rosb}
    \dfrac{d }{dv} \beta(v) &= (B-\beta(v))\dfrac{G'(v)}{G(v)} \dfrac{v\beta(v) + \beta(v)^2\log(B-\beta(v))}{v(B-\beta(v))-\beta(v)^2}.
\end{align}
Similar to the first-order condition for ROI-maximizing bidders, this is a general non-linear ODE. But even for a uniform prior, we cannot solve this ODE analytically. 
As we will discuss, SODA provides a numerical technique to quickly converge to an equilibrium in all of these utility models.


\section{Learning Algorithms}
In the following section, we focus on the learning algorithms agents employ to increase their individual rewards. 
From an agent's point of view, the repeated interaction in an auction is an online optimization problem. 
At each stage $t = 1,2,\dots$ the agent chooses a strategy $x_t \in \Xcal $ from some set and gets a utility $u_t(x_t)$. 
The utility in each stage is determined by the opponents' current strategy, i.e., $u_t(\,\cdot\,) = u_i(\,\cdot\,, x_{-i,t})$. 
If we assume that $\Xcal$ is a closed convex subset of some vector space and $u_t$ are concave functions (or convex loss functions), this problem is known as an online convex optimization problem \citep{shalev-shwartzOnlineLearningOnline2011}.

A standard performance measure of algorithms generating a sequence of strategies $x_t$ in such a setting is the notion of (external) regret. It is defined by
\begin{equation}
    \text{Reg}(T) = \max_{x \in \Xcal} \sum_{t=1}^T u_t(x) - u_t(x_t)
\end{equation}
and denotes the difference between the aggregated utilities after $T$ stages and the best fixed action in hindsight. We say that an algorithm has \textit{no regret} if the regret $ \text{Reg}(T)$ grows sublinearly. 
Note that this notion of no regret does not make any assumptions on the distribution of the utilities. 
It is, therefore, especially well suited in a multi-agent setting. 
It is also well known that if all agents follow a no-regret algorithm in multi-agent settings such as our repeated auction, the time average of the resulting strategies converges to a \textit{coarse correlated equilibrium (CCE)} \citep{blum2007ExternalInternalRegret}. There are many no-regret learners such as follow-the-regularized-leader (FTRL), where the agent plays a regularized best response to the aggregated utility function \citep{shalev-shwartzOnlineLearningOnline2011}. If one considers linear surrogates of the (concave) utility function, methods such as dual averaging \citep{nesterov2009PrimaldualSubgradientMethods} can be derived from FTRL.

While FTRL requires knowledge of the functions $u_t$ \textit{(full information feedback)} to compute an update, gradient-based methods use only the gradient $\nabla u_t(x_t)$ evaluated at the $x_t$ to compute the linear surrogate and update the strategy (\textit{gradient feedback}). 
Algorithms with \textit{bandit feedback} require even less and rely solely on the observed utility $u_t(x_t)$ of the played strategy $x_t$ in each stage $t$.
In our experiments, we focus on algorithms in the bandit model applied to repeated auctions with finitely many actions or bids. 
The bandit model mimics the type of information that real-time bidding agents get in display ad auctions after bidding. SODA is an exception as it uses gradient feedback. We primarily use SODA as a numerical tool to provide us with equilibrium bidding strategies to compare against. 

\subsection{Exp3-Algorithm}

If we consider an auction where agents have fixed valuations (complete information), the online optimization problem can be modeled as a multi-armed bandit where each arm describes one out of the possible discrete bids $b \in \Acal^d = \{b_1,\dots,b_m\}$. In general, multi-armed bandit problems are reinforcement learning problems with a single state, where a learner repeatedly chooses an action from a set of available actions and receives a reward associated with the chosen action  \citep{lattimore2020bandit}. An agent repeatedly pulls an arm of a slot machine and aims to maximize the cumulative reward. 
Well-known bandit algorithms include $\varepsilon$-greedy, UCB (Upper Confidence Bound), or Exp3 (Exponential-weight algorithm for Exploration and Exploitation). 
We will focus on Exp3 as a representative bandit algorithm but also report results for other methods.

\begin{algorithm}[h]
\SetKwInOut{Input}{Input}\SetKwInOut{Output}{Output}
\Input{exploration parameter $\varepsilon$, learning rate $\alpha$, number of rounds $T$}
Initialize strategy $x_1 = \tfrac{1}{m} (1, \dots, 1)$\; 
\For{$t=1,2,\dots,T$}{
    Choose action $b_t \sim P_t = (1-\varepsilon)\, x_t + \varepsilon\, \tfrac{1}{m} \textbf{1}_m$\;
    Receive reward $R_t$\;
    Update strategy $x_{t+1}(b) = \dfrac{x_t(b) \exp(\alpha R_t(b)))}{\sum_j x_t(b) \exp(\alpha R_t(b)))}$ for all $b \in \Acal^d$\;
}
\caption{Exp3-Algorithm}
\label{alg:bandit_alg}
\end{algorithm}
One key idea of Exp3 is to use an estimator for the reward of actions not played at time $t$.
We focus on a specific estimator known as an importance-weighted estimator. Let $b_t \in \Acal^d $ be the action played in iteration $t$ and $r_t \in [0,1]$ the respective reward. 
Then, the importance-weighted estimator $ R_t$ is defined by
\begin{equation}
    R_t(b) = \chi_{b_t}(b) \dfrac{r_t}{P_t(b_t)}, \quad \forall b \in \Acal^d
\end{equation}
where $P_t(b_t) $ is the probability of playing action $b$ in stage $t$,  $x_t \in \Xcal := \Delta \Acal^d$ the current mixed strategy, and $\chi$ the indicator function.
Note that $ R_t(b) $ is an unbiased estimate of $r_t$ for action $b$ conditioned on the history. This estimate is then used to update the strategy similar to the multiplicative weights update. 
This variant of Exp3 \citep{braverman2017selling} (see Algorithm \ref{alg:bandit_alg}) has an additional exploration probability, which we can also find in the bandit version of the online stochastic mirror descent from \citet[Sec.~31]{lattimore2020bandit}.
It is well-known that the Exp3 algorithm is a no-regret learner for $\varepsilon = \alpha \sim T^{-\beta}$ with $\beta \in (0,1)$ \citep{auer2002bandit}. 

\subsection{Q-Learning} \label{sec:qlearner}

Q-learning is the most well-known reinforcement learning algorithm. This is probably one reason, why it is predominantly used in the literature on algorithmic collusion \citep{Calvano.2020}, and specifically by \citet{banchio2022artificial}. 
In contrast to multi-armed bandit algorithms, Q-learning can handle different states of the world in which other actions might be optimal. Note that in \citet{banchio2022artificial} there is only one state and Q-learning reduces to an algorithm for a multi-armed bandit problem, a simple version of reinforcement learning. 

In Algorithm \ref{alg:q_learning_auction}, the input includes the discount factor $\gamma$, the learning rate $\alpha$, and the number of episodes $T$. The output is the learned Q-values for all state-action pairs. The Q-table is initialized with some initial value for all state-action pairs. The Q-learning update rule updates the Q-table based on the observed reward and the new state. The exploration-exploitation strategy used to select actions should be chosen carefully to balance exploration with exploitation. 
Similar to \citet{banchio2022artificial}, we focus on the $\varepsilon$-greedy rule, where the agent takes the action that maximizes the Q-value with probability $1-\varepsilon$ (exploitation) or takes an action at random with probability $\varepsilon$ (exploration). This allows us to compare our results to theirs. 
The exploration probability decays over time, e.g., according to $\varepsilon = \epsilon \exp ( - \beta t)$.
\begin{algorithm}[ht]
\SetKwInOut{Input}{Input}\SetKwInOut{Output}{Output}
\Input{Discount factor $\gamma$, learning rate $\alpha$, number of rounds $T$,\newline exploration parameter ($\epsilon,\, \beta$), initial Q-values $Q_0$}
Initialize Q-values $Q(b) = Q_0$ for all $b \in \Acal^d$\;
\For{$t=1,2,\dots,T$}{
    Set $\varepsilon = \epsilon \exp ( - \beta t) $\;
    Choose action $b_t \in \argmax_{b \in \Acal^d} Q(b)$ with probability $1-\varepsilon$ or at random with probability $\varepsilon$\;
    Receive reward $r_t$\;
    Update Q-value $Q(b_t) = (1-\alpha) Q(b_t) + \alpha \left[ r_t + \gamma \max_b Q_t(b) \right] $\;
    }
\caption{Q-Learning Algorithm without States}
\label{alg:q_learning_auction}
\end{algorithm}
Note that in general, the continuation value is the maximal Q-value of the following state $s_{t+1}$, i.e., $\gamma \max_b Q_t(b, s_{t+1})$. Instead, similar to \citet{banchio2022artificial}, we consider a simplified version without states which reduces the Q-table to a Q-vector. Our baseline implementation is also \textit{optimistic} Q-learning, where the Q-vector is initialized with values larger than the maximal possible discounted payoff. The purpose of such an initialization is to ensure that the algorithm won’t stop experimenting until all of the Q-values have sufficiently decreased. \citet{banchio2022artificial} argue that in their multi-agent setting, the advantage offered by optimism is a phase of experimentation at the beginning, which improves convergence. We find that the initialization is crucial and also explore alternative initializations in our experiments. 

\subsection{Extensions to Model Incomplete Information} \label{sec:pop_game}

The algorithms described so far were based on the model of complete-information games where the values of the bidders are publicly known. 
In practice, bidders do not have complete information, and algorithms need to take into account incomplete information.
This can be done by discretizing the type or valuation space $\Vcal^d := \{v_1, \dots, v_n \} \subset \Vcal$ of bidders with a corresponding prior distribution $F^d$. 
The discrete spaces can be constructed from the continuous game $\Gcal$ using a subset of equidistant points from the spaces and a numerical integration rule to derive a discretized prior $F^d$.
Overall, this defines a discrete incomplete-information game $\Gcal^d = (\Ical, \Vcal^d, \Acal^d, u, F^d)$ for the standard continuous Bayesian game as it is described in textbooks \citep{krishna2009auction}. Note that in practice every auction is discrete, because bidders can only submit bids up to a specific number of decimals.

The standard bandit algorithms are not defined for games with incomplete information. However, an interpretation as a population game \citep{hartline2015NoRegretLearningBayesian} allows for an extension that can be readily implemented. 
For this, we consider $N$ populations, where each population $i \in \Ical$ consists of $n = \vert \Vcal^d \vert$ players. All players within a population have different types $v \in \Vcal^d$. We write $v_i$ the player from population $i$ with valuation $v \in \Vcal^d$. In each round of the population game, nature draws a player from each population according to the prior $F^d$, i.e., $v_i \sim F^d$ for all populations $i \in \Ical$, who play against each other. Since this game is a complete-information game with finite actions, we can directly apply the algorithms described previously and learn a strategy for each $v_i$ of some population $i$.
Aggregating the strategies over all $v_i$ gives us a strategy for agent $i$ in the discrete incomplete-information game. Note that in this setting, the algorithms have to learn to perform well against the ``average agent'' of the opponents' populations.

\subsection{SODA} \label{sec:soda}
Let us also describe a gradient-based method, the SODA learning algorithm \citep{bichler2023soda}. 
This method provides us with an equilibrium strategy to compare against in situations where no analytical solution is available. If SODA converges to a strategy profile, it has to be an equilibrium. 

The SODA learning algorithm also acts on the discretized auction game $G^d$. The main idea is to use distributional strategies in the discretized auction game. Distributional strategies are an extension of mixed strategies in complete-information games to settings with incomplete information \citep{Milgrom1985}. 
They can be represented by matrices $s \in \R^{n \times m}$, where each entry $s_{ij}$ of this matrix represents the probability of a valuation-action pair $(v_i, b_j) \in \Vcal^d \times \Acal^d$. 
To be consistent with the auction game, the distributional strategy has to satisfy the marginal condition  $\sum_{j=1}^m s_{ij} = F^d(v_i)$, i.e., the probability of a value over all actions has to be equal to the probability of the value given by the prior. We denote the set of such distributional strategies by $\Scal$. 
Given a profile of such discrete distributional strategies, one can compute the expected utility $ \tilde u_i $ as sum over all outcomes weighted by the probabilities induced by $ s = (s_1,\dots, s_N) $.
If we consider an example with $N=2$ agents, we get an expected utility for agent 1 by
\begin{equation} \label{eq:exp_util_distr}
    \tilde u_1(s_1, s_2) = \sum_{i_1,j_1 = 1}^{n,m} (s_1)_{i_1j_1} \sum_{i_2,j_2 = 1}^{n,m} (s_2)_{i_2j_2} \, u_1(b_{j_1}, b_{j_2}, v_{i_1}),
\end{equation}
where $u_1(b_1, b_2, v_1)$ is the ex-post utility of agent 1. 

The expected utility, together with the set of distributional strategies, allows us to interpret the auction game as a complete-information game $ \Gamma = (\Ical, \Scal, \tilde u) $. 
Note that the set of discrete distributional strategies is a compact and convex subset of $ \R^{n \times m} $ and the expected utility is differentiable and linear in the bidder's own strategy. This allows us to rely on standard tools from online convex optimization to compute the NE of $ \Gamma $, which corresponds to a BNE of $ \Gcal^d $.
 
In our experiments, we use dual averaging \citep{nesterov2009PrimaldualSubgradientMethods} with an entropic regularization term (known as exponentiated gradient ascent), mirror ascent \citep{nemirovskij1983problem} with an Euclidean mirror map (standard projected gradient ascent), and the Frank-Wolfe algorithm \citep{frank1956}.
\begin{algorithm}[h]
    \SetAlgoNoLine
    \KwIn{Initial distributional strategy $s_1$, sequence of step sizes $\{\eta_t\}_{t=1}^T$}
    Initialize $y_{i,t} = 0$ for all $i \in \Ical$
    \For{$t = 1, 2, \dots, T$}
    {
        \For{each agent $i\in \mathcal I$}
        {
            observe gradient $c_{i,t} = \nabla_{s_i} \tilde u_i(s_{i,t},s_{-i,t})$\;
            update dual variable $y_{i,t+1} = y_{i,t} + \eta_t c_{i,t}  $\;
            update strategy $s_{i,t+1} = \argmin_{s \in \Scal^d} \Vert s - y_{i,t+1} \Vert_2^2 $\;
        }
    }
    \caption{Simultaneous Online Dual Averaging (SODA)}
    \label{algo:soda}
\end{algorithm}
Given the respective method, all agents simultaneously compute their individual gradients and perform the corresponding update step, as described in Algorithm \ref{algo:soda} for dual averaging. 
Specifically for dual averaging, one can show that if this procedure converges to a single point $s \in \Scal^d$, then this strategy profile is a NE in the approximation game, and thereby a BNE for the discretized auction game \citep[Corollary~1]{bichler2023soda}. Moreover, for some single-object auction formats, such as first- or second-price sealed-bid and all-pay auctions with quasi-linear utility functions, it is shown that if SODA finds an approximate equilibrium of the discretized auction game, this is also an approximate equilibrium of the continuous game \citep[Theorem~1]{bichler2023soda}. Therefore, SODA provides an ex-post certificate. 

Note that SODA serves only as a baseline as it allows us to approximate the BNE quickly and with high accuracy. It is not meant as an algorithm to simulate display ad auctions. While bandit algorithms such as Exp3 only require information about the price and whether they won or lost in an auction, in SODA the auctioneer would need all agents' mixed strategies to provide the respective feedback for the agents. 


\section{Results}
In this section we report the results of learning dynamics under different learning algorithms, different informational assumptions, and different utility models, to analyze revenue in repeated first-price auctions compared to their second-price counterparts.

\subsection{Experimental Design} \label{sec:exp_design}
In our experiments we simulate independent agents bidding in repeated first- and second-price auctions. Each experiment consists of $T$ iterations and is repeated several times. 
We focus on experiments with bandit feedback, where in each iteration $t \in \lbrace 1, \dots, T \rbrace$, agents submit their bid from some discrete action space $b_t \in \Acal^d \subset \Acal = [0,1]$, observe their utility (bandit feedback), and update their bidding strategies according to some learning algorithm. The objective of the agents is to maximize their utility, which depends on the bids of all agents and on their own valuation $v_t \in \Vcal^d \subset \Vcal = [0,1]$ from some discrete type space.

We will first discuss the \textit{complete-information model} where the agents have a fixed valuation $v_t=1$. We are interested in comparing the performance of different learning algorithms in the first- and second-price auction with payoff-maximizing agents, i.e., quasi-linear utility functions. We all experiments for $T=5\cdot 10^5$ iterations and repeat each 100 times.

In the \textit{incomplete-information model}, the valuations of each agent are drawn uniformly (i.i.d.) in each iteration, i.e., $v_t \sim U(\Vcal^d)$. Using the population game interpretation described in Section \ref{sec:pop_game}, we learn a separate strategy for each of the agent's valuations, which allows us to use the algorithms from the complete-information setting. 
Since each agent-value pair gets less feedback that way, we run these experiments for $T = 10^7$ iterations and repeat them 10 times. The experiments differ not only in the payment rule but also in the utility function of the agents, namely quasi-linear (QL), ROI, and ROSB. For ROI and ROSB to be well defined, we introduce a reserve price $r=0.05$, which we also use in the QL setting for better comparison. For ROSB, we assume a budget of $B = 1.01$ but provide sensitivity analyses in the appendix.

\subsection{Complete-Information Models} \label{sec:results_compl_info}

In our first set of experiments, we focus on the complete-information model by \citet{banchio2022artificial}.  

\paragraph{Setting} Two agents $i \in \{1,2\}$ bid in a first- and second-price auction. All agents value the item at $v = 1$ and choose bids from $\Acal^d = \{0.05, 0.10, \dots, 0.90, 0.95 \}$. 
For the first-price auction, it can be shown that with two agents and a random tie-breaking rule, the NE is to bid $b_1 = b_2 = 0.95$ or $b_1 = b_2 = 0.90$ for both of them. With more than 2 bidders, the NE is to only bid $b_i =0.95$ for all $i \in \{1, \dots, N\}$. In the second-price auction, the unique NE is $b_1 = b_2 = 0.95$.

\paragraph{Results}
Below, we report on experiments with Exp3 and Q-learning with the parameters as defined in Table \ref{tab:param_compl_info}.

\input{tables/table_learner_compl_info}

First, we consider two symmetric agents using variations of the Q-learner considered by \citet{banchio2022artificial}, which serves as a baseline. \citet{banchio2022artificial} used an optimistic initialization for their learners, which we replicate by setting the initial Q-values to $\tfrac{1}{1-\gamma}$. Optimistic initialization means that at the start for every state-action pair, the value of Q is larger than the maximum payoff it could ever be achieved. 
We can reproduce their results but observe that parameter changes lead to higher prices compared to this baseline Q-learner (see Figure \ref{fig:q_learner}).

\begin{result}
    The collusive outcomes by using Q-Learning in the complete-information setting from \citet{banchio2022artificial} rely on the parameters of the Q-values. We find evidence that optimistic initialization enables collusive outcomes while lower initializations or lower discount rates lead to higher prices closer to the equilibrium. 
\end{result}

Without the decreasing exploration rate (fix $\varepsilon$), we see cyclic behavior where the Q-learner often returns to low prices but fails to converge to a fixed policy. 
Even more interesting is the observation that initializing the Q-values with zero (and fixed exploration rate) instead of optimistic leads to a completely different outcome where the learner converges to the Nash equilibrium. 
One possible explanation is that collusive prices are more stable for higher Q-values since the rewards are much higher. 
This way, the learning algorithm might get stuck with lower bids. 
Higher actions are played if we remove this hidden preference for low actions by initializing them differently. 
It is also interesting to note that if we stick to optimistic initialization but have a small discount factor $\gamma$, the learners also converge to higher prices closer to the equilibrium. 
Revenue also increases again when agents use different parameters (see Figure \ref{fig:compl_info_fp_revenue_a}) or compete with more bidders.

\begin{figure}[h]
	\begin{center}
	\begin{subfigure}{0.45\textwidth}
		\includegraphics[width=0.9\textwidth]{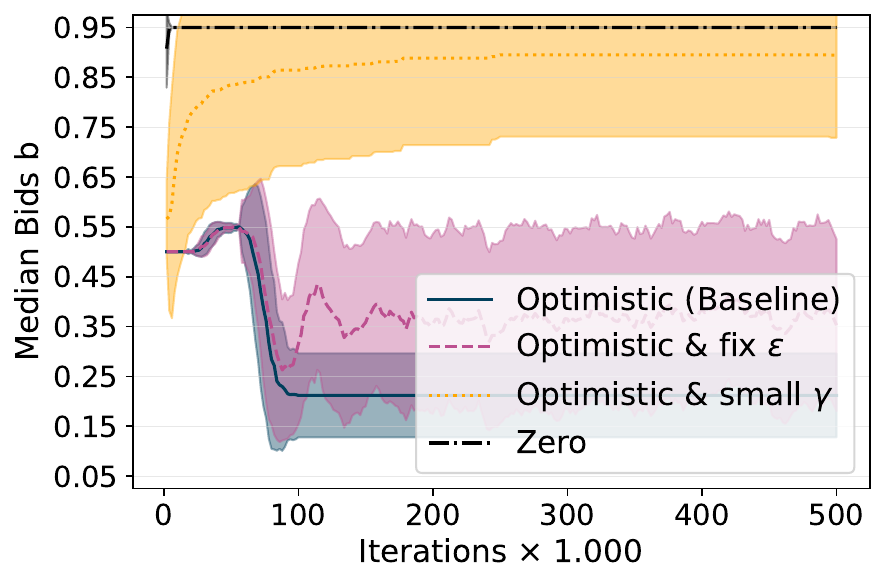}
		\caption{Q-Learner vs. Q-Learner (Symmetric Agents)}
	\end{subfigure}
	\begin{subfigure}{0.45\textwidth}
		\includegraphics[width=0.9\textwidth]{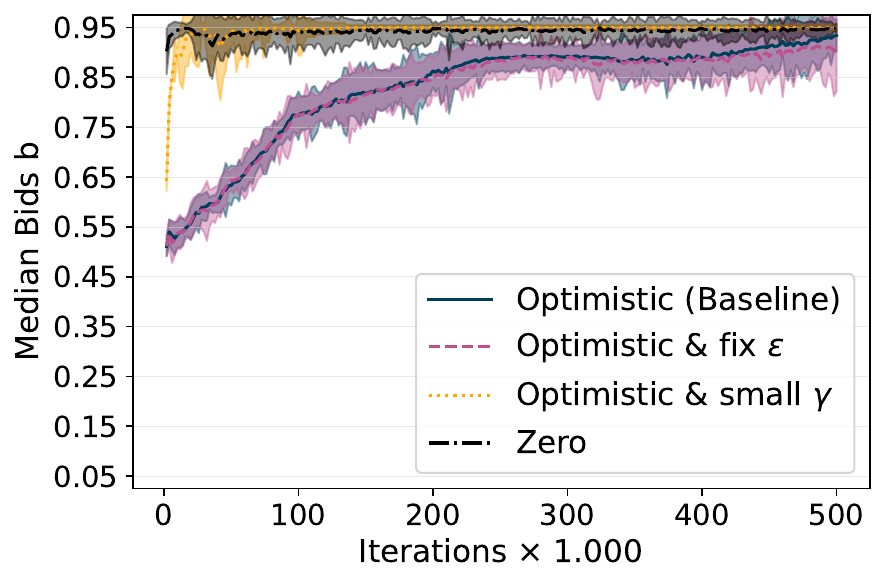}
		\caption{Q-Learner vs. Exp3 (Asymmetric Agents)}
	\end{subfigure}
	\caption{Median Bids in the FPSB Auction with Two Agents using Different Q-Learners}
	\label{fig:q_learner}
	\end{center}
	\footnotesize
	We split the learning process into intervals of two thousand iterations and compute the median bid of an agent for each interval. The shaded area shows the mean $ \pm $ std of these median bids over 100 runs, i.e., repetitions of the experiment. In the first plot, we show experiments for different parameters of the Q-learner used by both agents (symmetric). In the second plot, we show the results for agents using different learning algorithms (asymmetric), i.e., Q-learner with different parameters against an agent using Exp3. Note that we report the median instead of the average bid since the exploration makes the average bid harder to interpret, especially when we converge to high bids.
\end{figure}

\begin{result}
    If we use Exp3 and Q-learning algorithms, Q-learning becomes more robust and agents bid significantly higher compared to the symmetric case with identical Q-learners. Similar results can be observed when other combinations of bandit algorithms compete in a first-price auction.
\end{result}

The second plot in Figure \ref{fig:q_learner} shows additional experiments where the second agent uses the Exp3 algorithm, which we use as an example for no-regret bandit algorithms. If both agents use Exp3, they converge to a Nash equilibrium quickly. If Q-learners compete against Exp3, we observe significantly higher bids compared to the symmetric settings in the first plot. The Q-learner converges to a Nash equilibrium, even though it takes more iterations.

While we focused on Exp3 so far, similar results hold for \textit{other popular bandit algorithms} such as $\varepsilon$-Greedy, Thompson Sampling, and UCB1. Most experiments with different combinations of algorithms converged to the Nash equilibrium. Only UCB1 and $\varepsilon$-Greedy can lead to prices below the pure Nash equilibria for some Q-Learners (Figure \ref{fig:compl_info_fp_revenue}).

\begin{figure}[h] 
	\begin{center}
		\begin{subfigure}{0.45\textwidth}
			\includegraphics[width=0.9\textwidth]{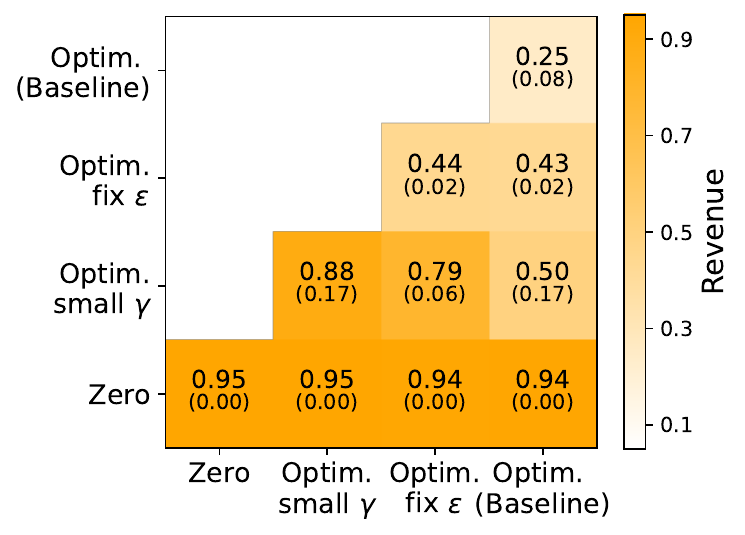}
			\caption{Q-Learner vs. Q-Learner}
			\label{fig:compl_info_fp_revenue_a}
		\end{subfigure}
		\begin{subfigure}{0.45\textwidth}
			\includegraphics[width=0.9\textwidth]{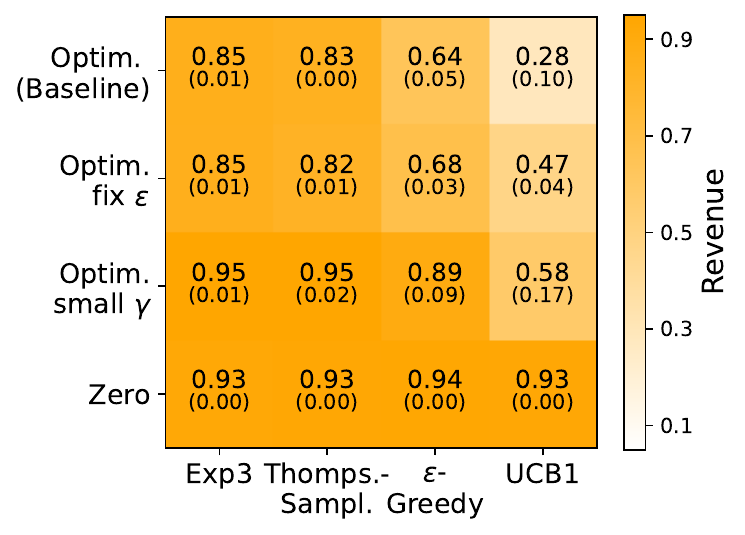}
			\caption{Q-Learner vs. Bandit}
			\label{fig:compl_info_fp_revenue_b}
		\end{subfigure}
		\caption{Revenue in the FPSB Auction with Two Agents Using Different Learning Algorithms}
		\label{fig:compl_info_fp_revenue}
	\end{center}
	\footnotesize
	We compute the average revenue over the last $10^5$ iterations (excluding an initial learning phase of $4 \cdot 10^5$ iterations) and report the mean and standard deviation (number in parentheses) over 100 runs of these experiments. 
    Each field represents the mean revenue in the FPSB auction with two competing agents using the respective learning algorithm indicated on the x- and y-axis. Revenues within the range of $[0.90, 0.95]$ are expected, if agents play according to the Nash equilibrium of the stage game.
\end{figure}
A short introduction to the additional bandit algorithms, together with experiments, can be found in Appendix \ref{app:bandit}.

\subsection{Incomplete-Information Model}
In the following experiments, we make the assumption that agents compete against a distribution of bidders with different valuations and study this Bayesian setting under different utility models. As described in Section \ref{sec:pop_game}, this corresponds to the classical incomplete-information setting known from auction theory.

\subsubsection{Equilibria Learned with SODA} \label{sec:results_soda}
Before we analyze the behavior of bandit algorithms in this setting, we use SODA on a finer discretization to get a good approximation of the equilibrium strategies since analytical solutions are not available for all settings (see discussion in Section \ref{sec:equi}).

\paragraph{Setting} 
The action and type space $\Acal = \Vcal = [0,1]$ of the $N=2$ agents are discretized using $n=m=64$ equidistant points. We consider the first- and second-price auction under different utility models, namely QL, ROI, and ROSB maximizing agents and different variations thereof.

\paragraph{Results}
To measure the quality of computed strategies, we use the \textit{relative utility loss $\ell$} as a metric. Given a distributional strategy $s_i$, $\ell$ measures the relative improvement of the expected utility we can get by best responding to the other agents strategies $s_{-i}$, i.e., 
\begin{equation} \label{eq:discr_util_loss}
        \ell(s_i, s_{-i}) = 1 - \dfrac{\tilde u_i(s_i^{br}, s_{-i})}{\tilde u_i(s_i, s_{-i})}.
\end{equation}
Note that in the discretized game, the best response $s_i^{br}$ is the solution of a simple LP. If the absolute utility loss $\ell \cdot \tilde u_i(s_i, s_{-i}) < \varepsilon$ for all agents, the computed strategy profile is a $\varepsilon$-BNE in the discretized auction game. 
In our experiments, we observe that on average $\ell \leq 10^{-3}$ in all instances, except for the asymmetric setting, where the relative utility loss is slightly higher for some combinations. Moreover, we can show that we closely approximate the analytical equilibrium strategies for QL and ROI (see Table \ref{tab:results_soda} in Appendix \ref{app:soda}).

\begin{figure}[h]
    \begin{center}
        \begin{subfigure}{0.32\textwidth}
            \includegraphics[width=0.9\textwidth]{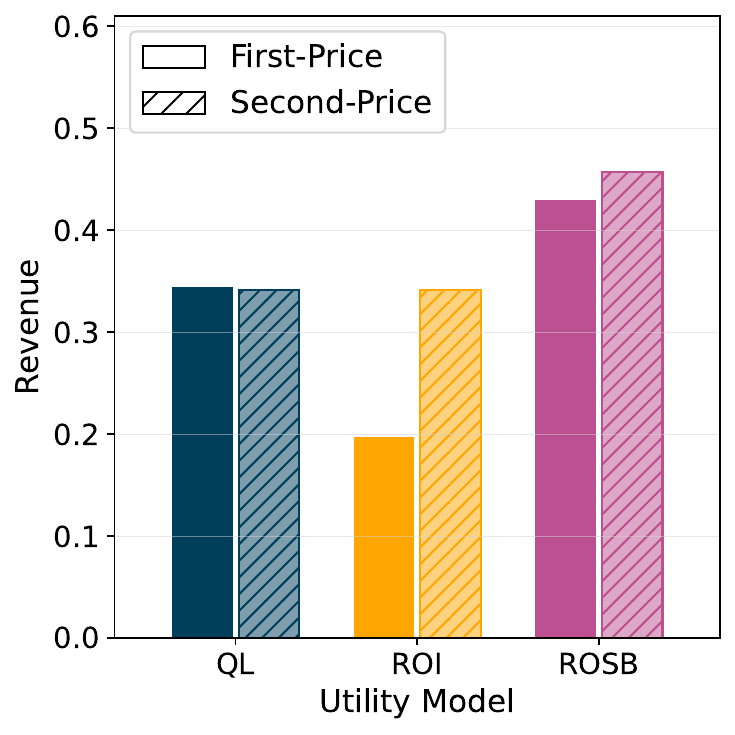}
            \caption{2 Bidders, Uniform Prior}
            \label{fig:soda_a}
        \end{subfigure}
        \begin{subfigure}{0.32\textwidth}
            \includegraphics[width=0.9\textwidth]{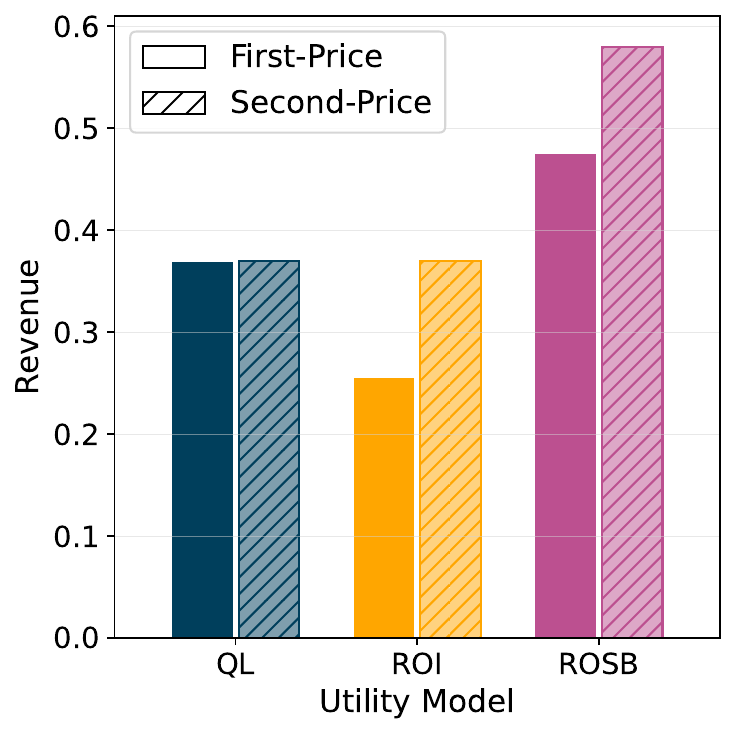}
            \caption{2 Bidders, Gaussian Prior}
        \end{subfigure}
        \begin{subfigure}{0.32\textwidth}
            \includegraphics[width=0.9\textwidth]{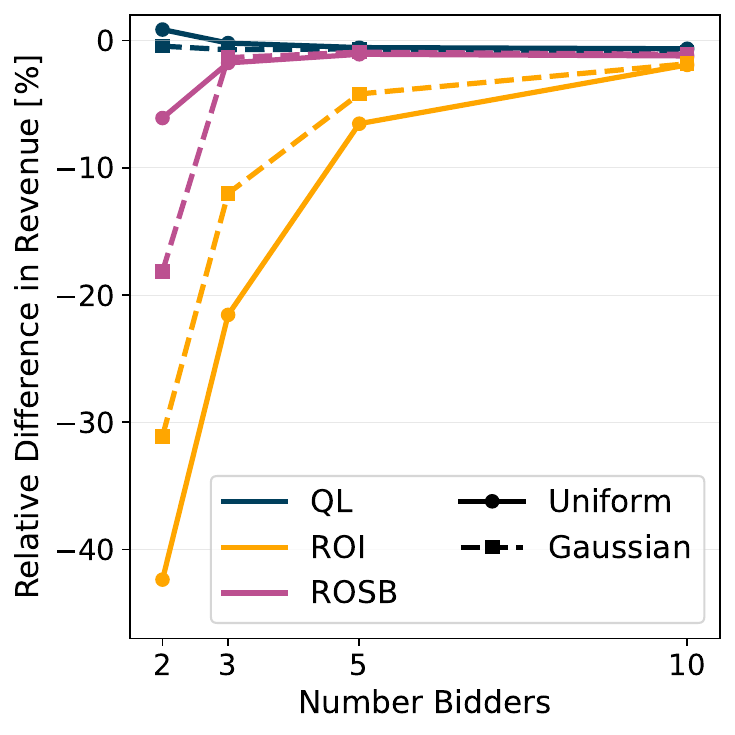}
            \caption{N Bidders}
        \end{subfigure}
        \caption{Expected Revenue in the FPSB and SPSB Auction with Different Utility Models.}
        \label{fig:soda}
    \end{center}
    \footnotesize
    We use the computed equilibrium strategies from SODA to compare the expected revenue between the first- and second-price auction for different utility models and prior distributions. In the first plot, the valuations are distributed uniformly, while we use a truncated Gaussian ($\mu=0.5, \, \sigma=0.3$) prior for the second plot. In the last plot, we visualized the relative revenue loss when switching from a second-price (SP) to a first-price (FP) auction, i.e., Revenue(FP)/Revenue(SP)-1 for both priors and different numbers of agents.
\end{figure}

Using the computed equilibrium strategies, we can simulate auctions to approximate the expected revenue. To that end, we sample $2^{22}$ valuations and the corresponding bids according to the strategies for each agent. This allows us to simulate auctions and approximate the expected revenue. We report the average expected revenue over 10 runs of the experiment. 

An important insight from this analysis (see Figure \ref{fig:soda}) is that revenue equivalence breaks for return-on-invest and return-on-spend (with budget) maximizing agents. The number of bids in ad auctions can vary significantly, but understanding this difference in auction formats is important for advertisers and ad exchanges alike. 


\begin{result}
    For ROI and ROSB the revenue in first-price auction is lower compared to the second-price auction in equilibrium. The difference in revenue shrinks with increasing levels of competition.
\end{result}
Note that this result is robust with respect to different variations of these utility models, i.e., convex combinations of ROI and ROIS or different budgets for all settings (see Appendix \ref{app:payoff_robustness}). 
Furthermore, we observe similar results for asymmetric agents, i.e., agents with different utility functions. 
\begin{figure}[h]
    \begin{center}
        \begin{subfigure}{0.45\textwidth}
            \includegraphics[width=0.9\textwidth]{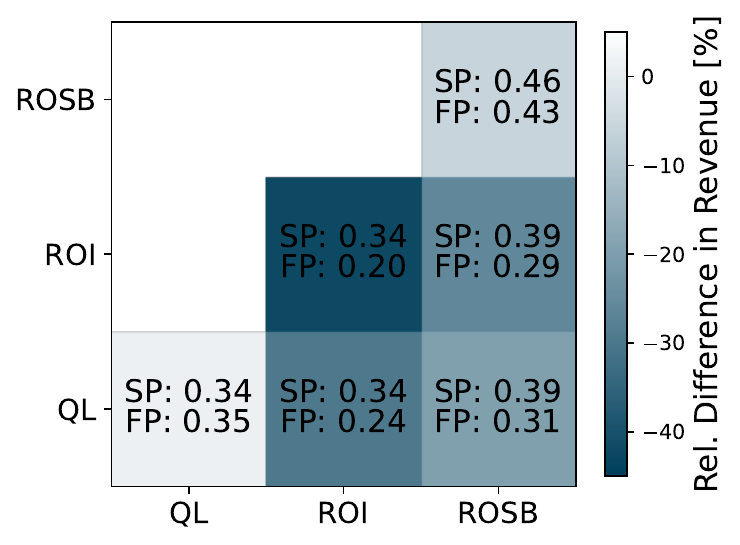}
            \caption{Uniform Prior}
        \end{subfigure}
        \begin{subfigure}{0.45\textwidth}
            \includegraphics[width=0.9\textwidth]{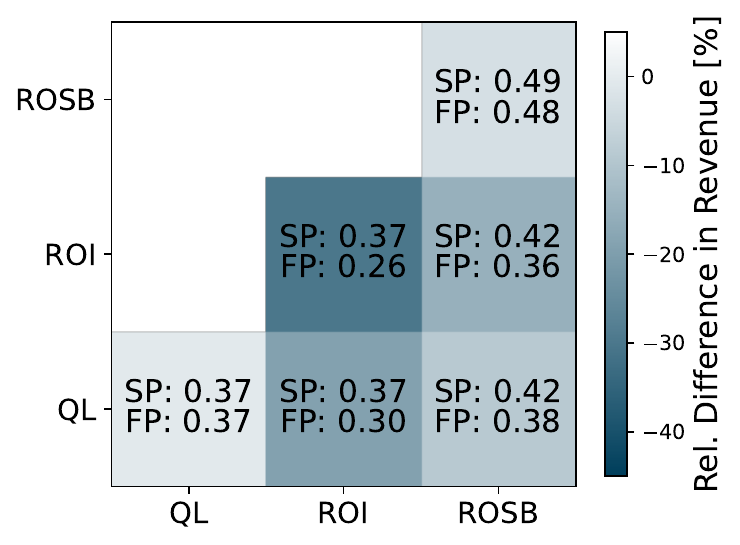}
            \caption{Gaussian Prior}
        \end{subfigure}
        \caption{Expected Revenue in the FPSB and SPSB Auction with Two Asymmetric Agents.}
        \label{fig:soda_asym}
    \end{center}
    \footnotesize
    We use the computed strategies from SODA to compare the expected revenue between the first- and second-price auction for agents with different utility models. The x-axis corresponds to the utility function of agent 1 and the y-axis to agent 2. The numbers denote the revenue in the first-price (FP) and second-price (SP) auction format for the respective combination of utility functions, while the color denotes the relative difference in revenue when switching from a second- to a first-price auction format. 
\end{figure}

We considered the expected revenue in first- and second-price auctions with two agents maximizing their payoff, return-on-invest, or return-on-spend (with budget). The revenues, as well as the relative differences between second- and first-price auctions, are visualized in Figure \ref{fig:soda_asym}. 
Interestingly, if one agent maximizes payoff but the other agent does not, the revenue can still decrease by 30\% (ROI) or 20\% (ROSB) when switching from a second-price to a first-price auction format with a uniform prior. The results also apply, albeit less severely, for these settings with a truncated Gaussian prior ($\mu = 0.5, \sigma = 0.3$).

\subsubsection{Bandit Algorithms} \label{sec:bandit_bayesian}
Real-world bidding agents only get bandit feedback. Therefore, we also analyze the revenue in the first-price auction with bandit algorithms (e.g., Exp3) in the ROI and the ROSB utility models. Since learning is much slower for learners with bandit feedback compared to SODA, we use a coarser discretization similar to the complete information setting.

\paragraph{Setting} We discretize the type and action space with $\Vcal^d = \Acal^d = \{0, 0.05, \dots, 0.95, 1\}$ which corresponds to $n=m=21$ discretization points. We focus on repeated first- and second-price auctions with $N=3$ agents\footnote{In this coarse discretization, bidding the reserve price for all valuations is the BNE for ROI and ROSB maximizing agents if we only consider $N=2$ agents. Therefore, we choose $N=3$.} and different assumptions on the utility function. As described in Section \ref{sec:pop_game}, each agent is represented by a population, which consists of different valuations. This means each agent (i.e. population) has $n=21$ different instances of the bandit learner, i.e., one for each valuation. To get comparable results to the complete information, we increase the number of iterations approximately by the number of discrete values, which leads to $10^7$ iterations instead of $5 \cdot 10^5$. 
In our experiments, we use Exp3 with a fixed exploration rate $\varepsilon=0.01$ and a fixed learning rate $\alpha=0.01$. For ROI and ROSB, we use a smaller learning rate $\alpha=0.0005$ as the magnitude of the utilities is larger. 

\paragraph{Results}

First, we observe that the revenues generated by the agents using Exp3 show a similar pattern to the expected revenues in equilibrium (see previous subsection). 

\begin{result}
    If we use Exp3 in the incomplete-information auction setting, the agents learn to bid according to the BNE and revenues for the first-price auction are lower compared to the second-price auctions for ROI and ROSB maximizing agents, as predicted by the equilibrium analysis.
\end{result}

We see that revenue equivalence does not hold for ROI and ROSB maximizing agents and especially for ROI maximizing agents, the revenue for the first-price auction is significantly lower compared to the second-price auction. 
Furthermore, the average revenue matches the values predicted by the BNE strategies in the discretized setting computed using SODA.
\begin{figure}[h]
    \begin{center}
        \begin{subfigure}{0.32\textwidth}
            \includegraphics[width=0.9\textwidth]{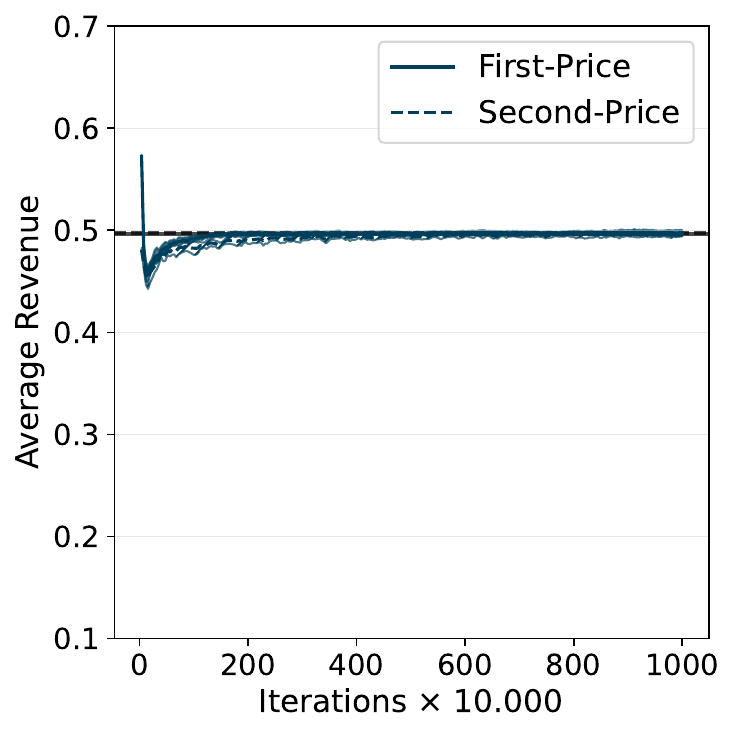}
            \caption{Payoff-Maximizing}
        \end{subfigure}
        \begin{subfigure}{0.32\textwidth}
            \includegraphics[width=0.9\textwidth]{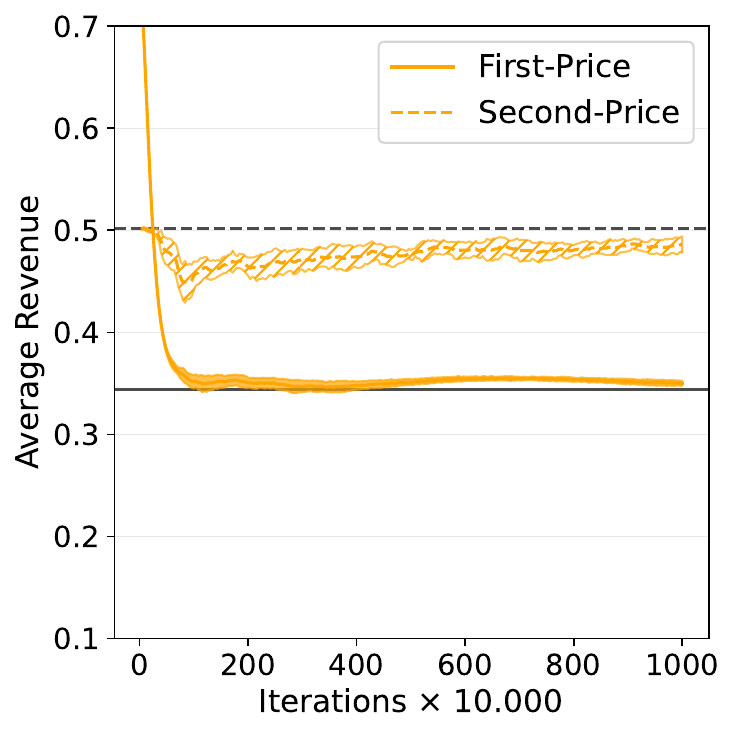}
            \caption{ROI-Maximizing}
        \end{subfigure}
        \begin{subfigure}{0.32\textwidth}
            \includegraphics[width=0.9\textwidth]{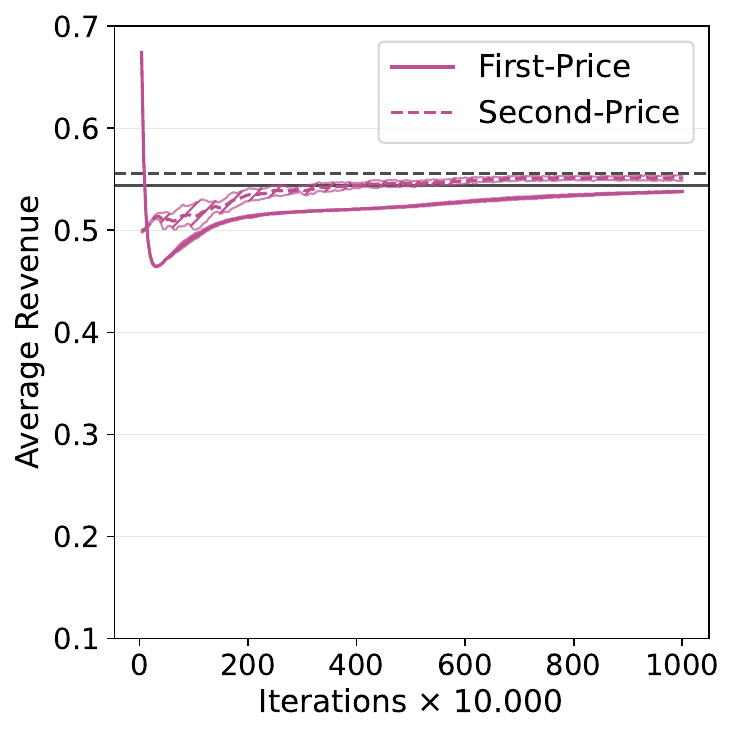}
            \caption{ROSB-Maximizing}
        \end{subfigure}
        \caption{Revenue in the FPSB and SPSB Auction with Three Bidders using Exp3 and Uniform Prior.}
        \label{fig:exp3_incomplete}
    \end{center}
    \footnotesize
    We run Exp3 for $10$ million iterations and compute the average revenue for all $40\thinspace000$ iteration intervals for payoff, ROI and ROSB maximizing agents. We plot the mean and standard deviations of this average revenue per interval over 10 runs. The black horizontal lines denote the expected revenue we would get with the BNE strategies computed using SODA.
\end{figure}

Not only is the average revenue close to the predictions from the equilibrium analysis, but we can also observe that the agents actually learn to play according to the Bayes-Nash equilibrium using a bandit algorithm such as Exp3. 
To this end, we look at the frequency of the bids within an interval and compare it to the distributional equilibrium strategy computed using SODA. 
As we can see in Figure \ref{fig:exp3_ql} - \ref{fig:exp3_roi}, the empirical frequency of actions played by Exp3 converges to the equilibrium strategies. 
More detailed results to the experiments, including different bandit learners and Q-learning, can be found in the Appendix \ref{app:bandit_bayesian}.

\begin{figure}[h]
	\begin{center}
    \includegraphics[width=1\textwidth]{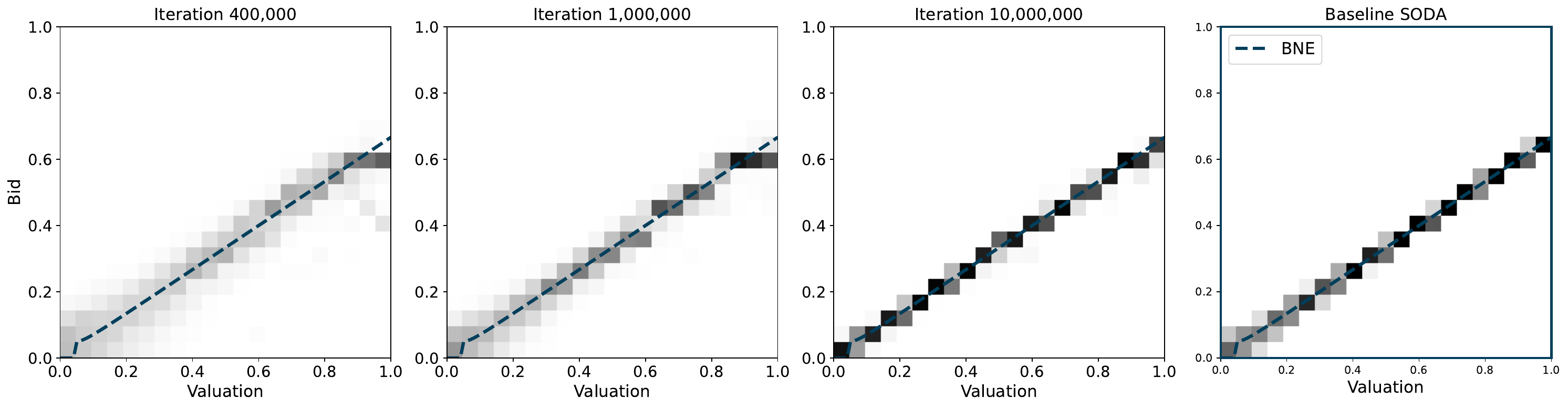}
	\caption{Strategies for Payoff-Maximizing Agents using Exp3 in the FPSB Auction with Uniform Prior}
	\label{fig:exp3_ql}
	\end{center}
	\footnotesize
	We run Exp3 for $10$ million iterations and visualize the frequency of the last $40\thinspace000$ bids for the respective valuations after $0.4$, $1$, and $10$ million iterations. On the last plot, we show the distributional equilibrium strategy computed with SODA. The colored lines denote the BNE in the continuous setting.
\end{figure}

\begin{figure}[h]
	\begin{center}
    \includegraphics[width=1\textwidth]{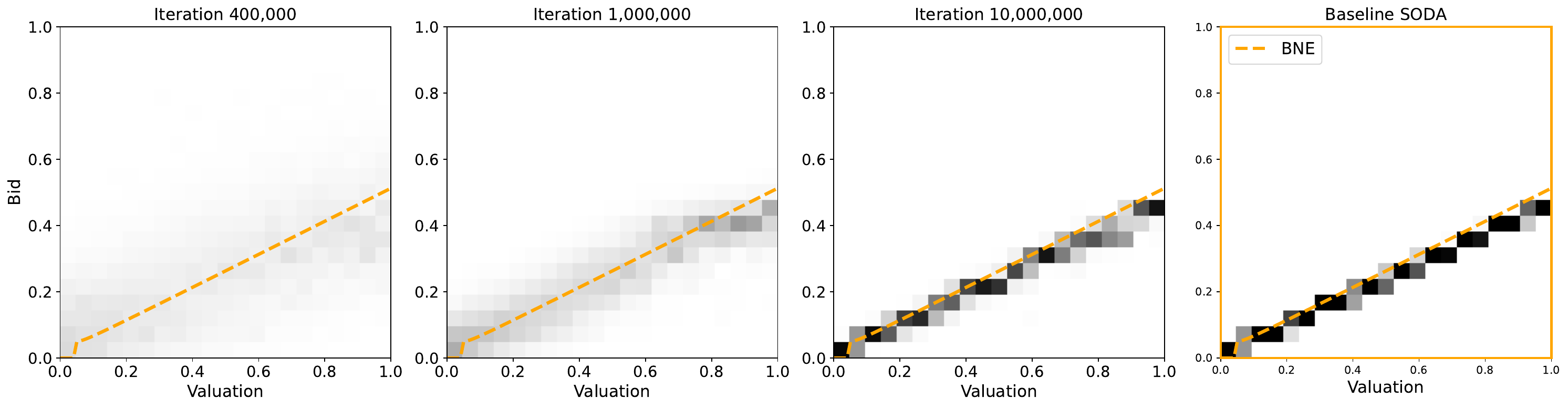}
	\caption{Strategies for ROI-Maximizing Agents using Exp3 in the FPSB Auction with Uniform Prior}
	\label{fig:exp3_roi}
	\end{center}
	\footnotesize
	We run Exp3 for $10$ million iterations and visualize the frequency of the last $40\thinspace000$ bids for the respective valuations after $0.4$, $1$, and $10$ million iterations. On the last plot, we show the distributional equilibrium strategy computed with SODA. The colored lines denote the BNE in the continuous setting.
\end{figure}

\section{Conclusions}
In this paper, we study the impact that a move from second-price to first-price auctions can have on the revenue of an ad exchange that faces learning bidding agents. While some exchanges report a revenue loss, the econometric analysis is difficult due to numerous confounding variables and changes in the market in recent years. Analytical models allow us to study the causal impact of such a policy change under ceteris paribus assumptions. However, standard models from auction theory do not adequately reflect the specifics of this market. 

First, auctions are conducted in milliseconds and the bidding is fully automated. Technically, multi-armed bandit algorithms are natural candidates for such tasks as they are fast and can adapt quickly to demand and supply in such environments. As a result, there is a growing literature on the use of such algorithms for bidding in online auctions. We want to understand the outcomes of such learning algorithms and whether they converge to the equilibrium strategies in a model or not. 

Second, ROI and ROS are widely used as objectives for Demand-Side Platforms, which is due to the fact that they serve as agents for an advertiser (the principal). They get some budget assigned and they are tasked to use this budget most effectively. As a result, ROI is a wide-spread metric to evaluate Demand-Side Platforms, and this is what they optimize. 
In the standard independent private values model in auction theory where payoff-maximizing bidders play their equilibrium bidding strategy, both auction formats can be shown to have the same revenue in expectation. This is not necessarily the case if the bidders maximize ROI. 

We show that wide-spread bandit algorithms all converge to equilibrium in complete- and incomplete-information models, with few exceptions. If the algorithms are mixed, we see convergence throughout. In practice, there might be changes in supply and demand over time, but if we do not observe algorithmic collusion in this model, it is unlikely to emerge in a more dynamic environment. These results are important in different ways. 
First, learning algorithms can cycle or even lead to chaotic dynamics in games. The fact that they converge to equilibrium in auction games without knowing the prior distribution a priori is important and makes game-theoretical equilibrium predictions very credible. While there is literature on learning in games, auctions have not been covered in this literature so far, and our paper provides new results to this literature stream. 
Second, our results suggest that if bidders maximize payoff in this model, then we can rely on the revenue equivalence theorem and the switch from the second- to the first-price auction should have been without loss for publishers and ad exchanges. 

The second feature of display ad auctions, i.e., the observation that ROI or ROS are widely used as objectives, leads to significant difficulties in equilibrium analysis. The equilibrium problem in auctions can be described as a system of non-linear differential equations, for which no exact solution theory is known. For ROI bidders, we can still derive an equilibrium bidding strategy for uniform distributions. For ROS bidders, even this is impossible. We draw on recent advances in equilibrium learning and compute equilibria for such auction models numerically. 
While these equilibrium learning algorithms allow us to certify an approximate equilibrium, we can show that also the bandit algorithms that are relevant in the practice of real-time bidding converge to equilibrium. Importantly, we show that in equilibrium, the revenue of the first-price auction is consistently lower than that of the second-price auction, independent of whether the agents are ROI or ROS maximizers. This is important to understand for publishers and ad exchanges alike. Overall, the move to first-price auctions was probably the most important policy change and the consequences for revenue are a fundamental and open problem.

As any other economic model, ours is also an abstraction from real-world practices. First, the bandit algorithms that we use in our work might differ from those implemented in practice. Apart from our own experience in working with Demand-Side Platforms, there is evidence in the literature for the use of bandit algorithms (see Section \ref{sec:intro}). Yet, the algorithms might have to consider additional detail in how they deal with different types of impressions and users or how they react to changes in supply and demand. Our model abstracts from these details, but we find that a large variety of algorithms converges to equilibrium, which is far from obvious. In particular, combinations of algorithms, the most likely scenario in the field, always converge. Our goal is not to replicate real-world exactly, which is impossible as companies will not reveal the details of their algorithms, but to see if under reasonable assumptions algorithmic collusion or other off-equilibrium outcomes emerge or not. The robustness of the results against such a large variety of bandit algorithms is remarkable. 

Second, also the objectives of the bidders in display ad auctions are confidential. Yet, there is a large literature suggesting that Demand-Side Platforms maximize ROI or ROS rather than payoff. Also here, we find a surprisingly robust result showing that agents maximizing ROI, ROS or combinations thereof lead to revenues in the first-price auction that are below those of the second-price auction. Importantly, learning agents converge to these equilibria as well. These results contribute to a central managerial question in the analysis and design of display ad auctions, but the techniques developed within this project also have implications beyond.

\section*{Acknowledgments}
This project was funded by the Deutsche Forschungsgemeinschaft (DFG, German Research Foundation) - GRK 2201/2 - Project Number 277991500 and BI 1057/9. The authors would also like to thank Laura Mathews, Markus Ewert and Fabian R. Pieroth for their support.



\bibliographystyle{informs2014} 



\newpage
\begin{appendix}
	\section{Additional Experiments for the Complete-Information Model} \label{app:bandit}
Additional to Exp3, we also considered other bandit algorithms to check for robustness of our observations.
We focus on three well-known bandit learners that handle the exploration-exploitation tradeoff differently (see for instance \cite{sutton2018ReinforcementLearningIntroduction, lattimore2020bandit}).
\begin{description}
    \item[$\varepsilon$-Greedy] This policy compares the current average reward of each action $a$, which is defined by
    \begin{equation}
        \hat R(a) = \dfrac{1}{N(a_t)}\sum_{s=1}^t \chi_a(a_t) r_s(a_t)
    \end{equation}
    with $\chi_a(a_t) = 1 $ if $a_t = a$ and $0$ else. 
    $N_t(a)$ denotes the number an action $a$ has been played up to iteration $t$.
    The $\varepsilon$-Greedy selects the action with the highest average reward with probability $1-\varepsilon_t$ and some random action with probability $\varepsilon_t$ ( $\varepsilon_t = 0.05$ in our experiments).
    \item[UCB1] Instead of exploring actions at random, UCB1 explores only actions with enough potential or uncertainty (upper-confidence-bound action selection). This is achieved by taking the action with the maximal value which consists of the average reward and an exploration bonus depending on the number an action has been visited. This value is given by
    \begin{equation}
        \hat R(a) + c \left(\frac{\ln t}{N_t(a)}\right)^{\tfrac{1}{2}},
    \end{equation}
    where $c$ is some exploration parameter ($c=4$ in our experiments).
    \item[Thompson-Sampling] Finally, we also use a Bayesian method first introduced in \citep{thompson1933}. Thompson sampling assumes a distribution $\Pcal_r(\theta_a)$ (with some parameter $\theta_a$) over the rewards of each action. In each iteration $t$, the learners sample parameters $\theta_a^t \sim \Pcal_\theta$ from some parameter distribution for each action $a$, play the action $a^*$ that maximizes the expected reward given the distribution $\Pcal_r(\theta_a^t)$, and update the parameter distribution $\Pcal_\theta$ for action $a^*$ given the observed reward $r_t(a^*)$.
    In our implementation, we assume that the rewards for each action follow a Gaussian distribution (with known variance $\sigma$) and thereby use a Gaussian prior for the parameter $\theta_a = \mu_a$ (conjugate prior).
\end{description}

\begin{figure}[h!]
    \begin{center}
        \begin{subfigure}{0.45\textwidth}
            \includegraphics[width=0.9\textwidth]{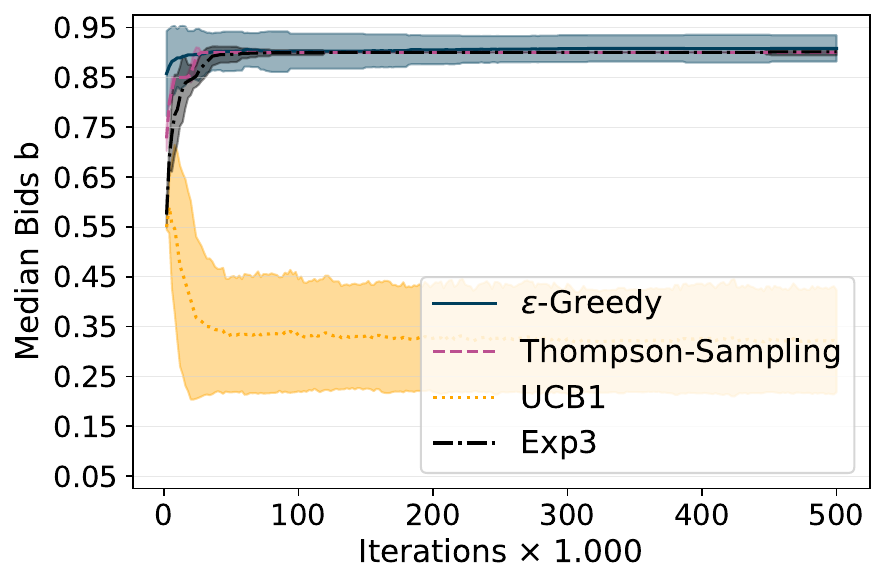}
            \caption{Bandit vs. Bandit (Symmetric Agents)}
        \end{subfigure}
        \begin{subfigure}{0.45\textwidth}
            \includegraphics[width=0.9\textwidth]{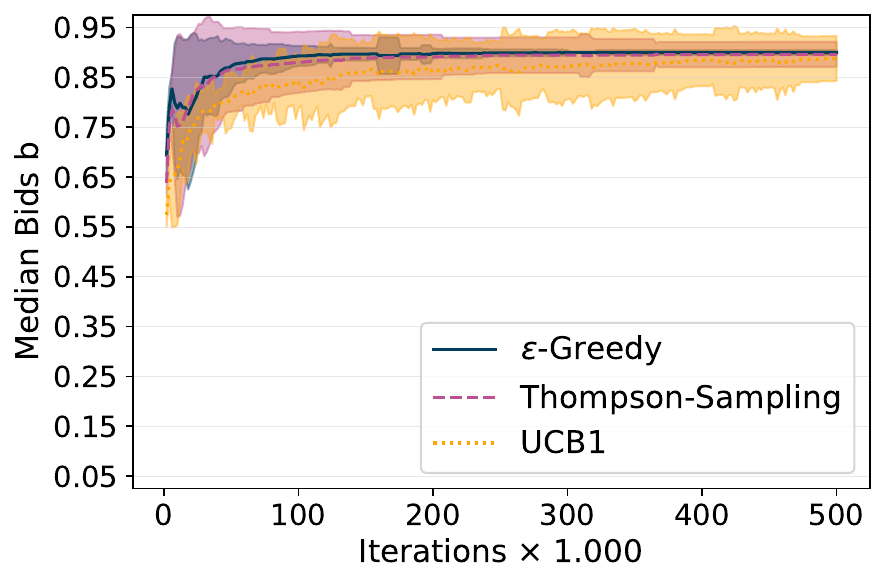}
            \caption{Bandit vs. Exp3 (Asymmetric Agents)}
        \end{subfigure}
        \caption{Median Bids in the FPSB Auction with Two Agents using Different Bandit Algorithms.}
        \label{fig:bandit}
    \end{center}
    \footnotesize
    We split the learning process of an agent into intervals of $2\thinspace000$ iterations and compute the median bid for each interval. The shaded area shows the mean $ \pm $ std of these median bids over 100 runs. The plots show the experiments for different bandit algorithms used by both agents (symmetric) and playing against an agent using Exp3 (asymmetric).
\end{figure}

Repeating the experiments with two players in a complete-information first-price sealed bid auction (Section \ref{sec:results_compl_info}), we observe that most of them converge to the Nash equilibrium (Figure \ref{fig:bandit}). 
The exception here is UCB1, which shows similar behavior to the Q-learner from \citet{banchio2022artificial}. 
But again, when playing against another bandit algorithm, such as Exp3, the bids are significantly higher.

If we compare the revenue under different payment rules for the Q-learner and different bandit algorithms (assuming symmetric agents), we can confirm the observations by \citet{banchio2022artificial} that the revenue in the first-price auction is significantly lower than in the second-price auction for their Q-Learner with optimistic initialization. 
But, our experiments further indicate that this might be specific to this learning algorithm. 
Overall, we observe that the first-price auctions might be harder to learn for some algorithms (i.e., some instances of Q-learning and UCB1). 
And since a failure to learn the Nash equilibrium automatically leads to lower revenue, these methods might seem to collude. But in fact, other versions of Q-learning (i.e., with zero initialization) show opposite results, and a number of reasonable bandit algorithms converge in both settings. 
Note that the slightly lower revenue (see Figure \ref{fig:compl_info_revenue}) in the first-price auction compared to the second-price auction for $\varepsilon$-Greedy, Thompson-Sampling, and Exp3 can be explained by the additional Nash equilibrium at $0.90$ for the first-price auction.
\begin{figure}[h]
    \begin{center}
        \begin{subfigure}{0.45\textwidth}
            \includegraphics[width=0.9\textwidth]{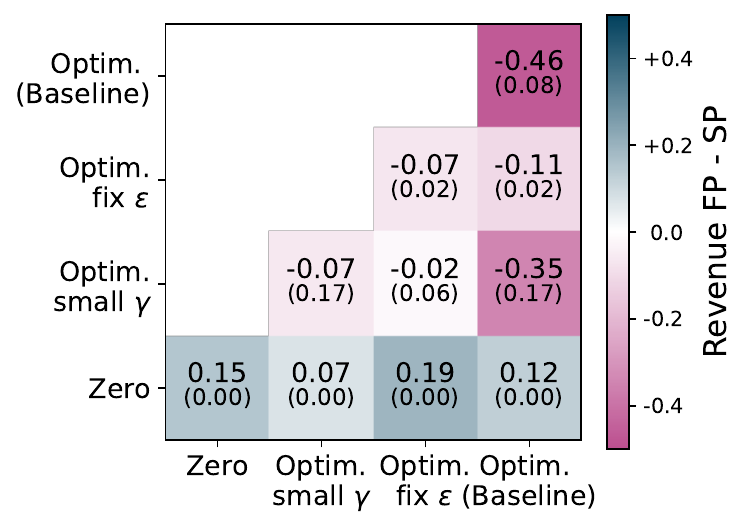}
            \caption{Q-Learner vs. Q-Learner}
        \end{subfigure}
        \begin{subfigure}{0.45\textwidth}
            \includegraphics[width=0.9\textwidth]{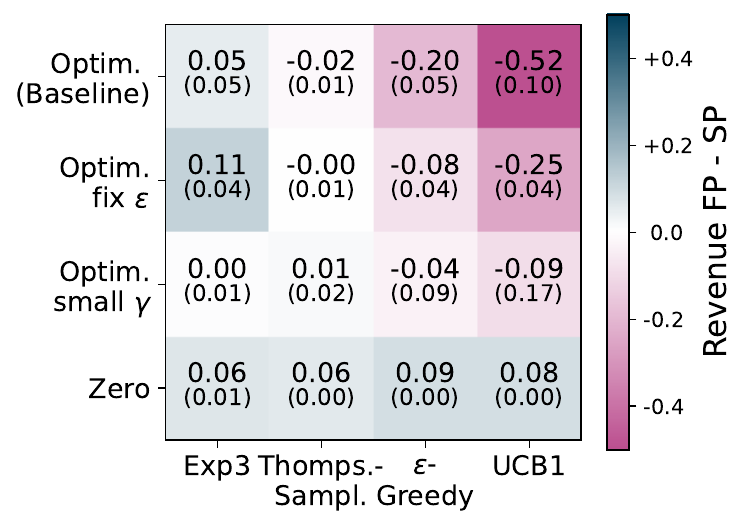}
            \caption{Q-Learner vs. Bandit}
        \end{subfigure}
        \caption{Difference in Revenue between the FPSB and SPSB Auction with Two Agents.}
        \label{fig:compl_info_revenue}
    \end{center}
    \footnotesize
    We compute the difference of average revenue between the first- and second-price auction over the last $10^5$ iterations (excluding an initial learning phase $4 \cdot 10^5$ iterations) and report the mean and standard deviation (number in parentheses) over 100 runs of these experiments. Each field represents the mean difference revenue two competing agents using the respective learning algorithm indicated on the x- and y-axis. Differences within the range of $[-0.05, 0.0]$ are expected, if agents play according to the Nash equilibrium of the stage games.
\end{figure}

\section{Additional Experiments for the Incomplete-Information Model} \label{app:incomplete_info}
\subsection{Evaluation of SODA against Analytical Equilibrium Strategies} \label{app:soda}
We applied SODA to the first- and second-price auctions with different utility models and quickly converge to a pure-strategy BNE. 
As described in Section \ref{sec:soda}, SODA acts on a discretized version we construct by discretizing the valuation and action space using $n=m=64$ equidistant points. 
Due to better performance, we use a tie-breaking rule in which both agents lose in ties for this discretization. The strategies are updated in each iteration using dual averaging with the entropic regularization term and a decreasing step size $\eta_t = 10 \cdot t^{-0.05}$. We stop the learning algorithm after $5\thinspace000$ iterations, or whenever the relative utility loss $\ell$ of all agents in the discretized game is less than $10^{-4}$. Each experiment is repeated $10$ times. To verify the computed strategies, we compare them to the analytical solutions in the settings where the BNE is known. To this end, we report the following metrics (similar to \cite{bichler2021npga, bichler2023soda}):
\begin{description}
    \item \textit{Relative Utility Loss} $\Lcal$. We approximate the expected utility using the sample-mean of the ex-post utilities, i.e., $ \tilde u_i (\beta_i, \beta^*_{-i}) \approx  \hat u_i (\beta_i, \beta^*_{-i}) := \tfrac{1}{n_v} \sum_{v} u_i( \beta_i(v_i), \beta_{-i}((v_{-i}))) $. The relative loss in the expected utility the agent receives when playing the computed strategy $s_i$ instead of the analytical equilibrium $\beta_i^*$
    \begin{equation} \label{eq:util_loss}
	\mathcal L_i(\beta_i)\ =\ 1 - \frac{\hat u_i(\beta_i, \beta^*_{-i})}{\hat u_i(\beta^*_i,\beta^*_{-i})},
    \end{equation}
     while all other agents play the equilibrium strategy $\beta^*_{-i}$.
    \item \textit{$L_2$ Distance}. The probability-weighted root mean squared error of $\beta_i$ and $\beta^*_i$ in the action space, which approximates the $L_2$ distance of these two functions:
    \begin{equation}
    	L_2(\beta_i)\ =\ \bigg(\frac{1}{n_\text{batch}}\sum_{v_i}\left(\beta_i(v_{i}) - \beta^*_i(v_{i})\right)^2\bigg)^\frac{1}{2}.
    \end{equation}
    This way we ensure that our computed strategy not only performs similarly in terms of utility but also approximates the actual known BNE.
\end{description}

\input{tables/table_soda}

To compute these metrics, we sample valuations (batch size $n_v = 2^{22}$), identify for each valuation the nearest discrete valuation  $v \in \Vcal^d$, and then sample a discrete bid $b \in \Acal^d$ from the computed strategy. The sampled bids are denoted by $\beta_i(v)$ in the definitions above. The same simulation procedure is used to approximate the expected revenue.
As Table \ref{tab:results_soda} shows, the computed strategies of the discretized game closely approximate the analytical BNE in the original continuous setting. 

\subsection{Additional Experiments for Combinations of ROI and ROS and Different Budgets} \label{app:payoff_robustness}
To check for the robustness of our observations, we also consider different variations of the different utility functions.
First, we look at a generalized version of the ROI and ROS utility models by considering a convex combination of both, which is defined by
\begin{equation} \label{eq:rois}
    u_i^{ROIS}(b, v_i) :=  (1-\lambda) u_i^{ROI}(b, v_i) + \lambda u_i^{ROS}(b, v_i)=  x_i(b) \dfrac{v_i - (1-\lambda) \, p_i(b)}{p_i(b)},
\end{equation}
with some parameter $ \lambda \in [0,1]$. Again, we apply SODA for different values of $\lambda$ and compare the revenue of the first- and second-price auction. The algorithm converges in all settings (utility loss $\ell <  10^{-4}$) and we use the computed strategies to simulate auctions and approximate the expected revenue. In the first plot of Figure  \ref{fig:soda_robust}, we can observe that except for ROS, where submitting the highest bid is the BNE for both payment rules, the revenue of the first-price auction is lower compared to the second-price auction.

\begin{figure}[h]
    \begin{center}
        \begin{subfigure}{0.32\textwidth}
            \includegraphics[width=0.9\textwidth]{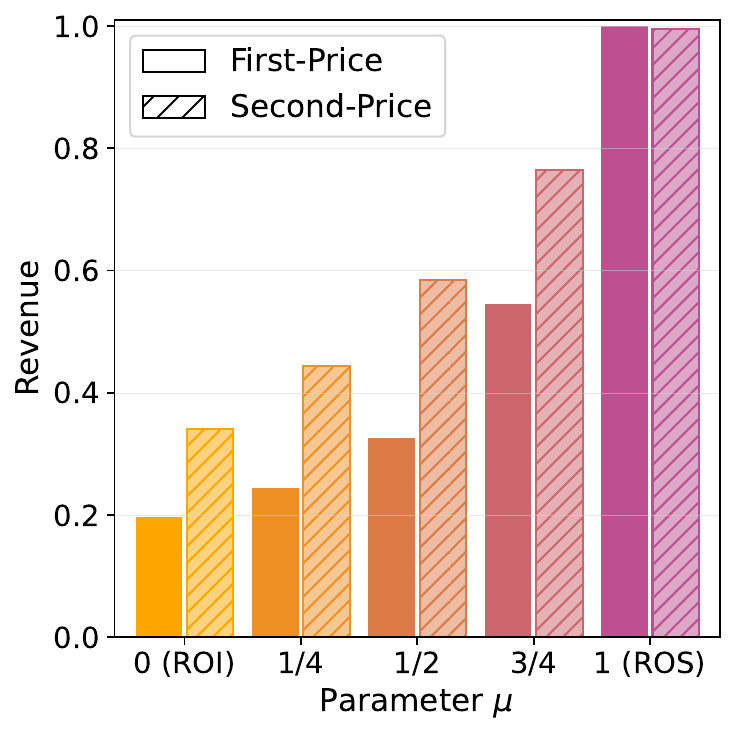}
            \caption{Combination of ROI and ROS}
        \end{subfigure}
        \begin{subfigure}{0.32\textwidth}
            \includegraphics[width=0.9\textwidth]{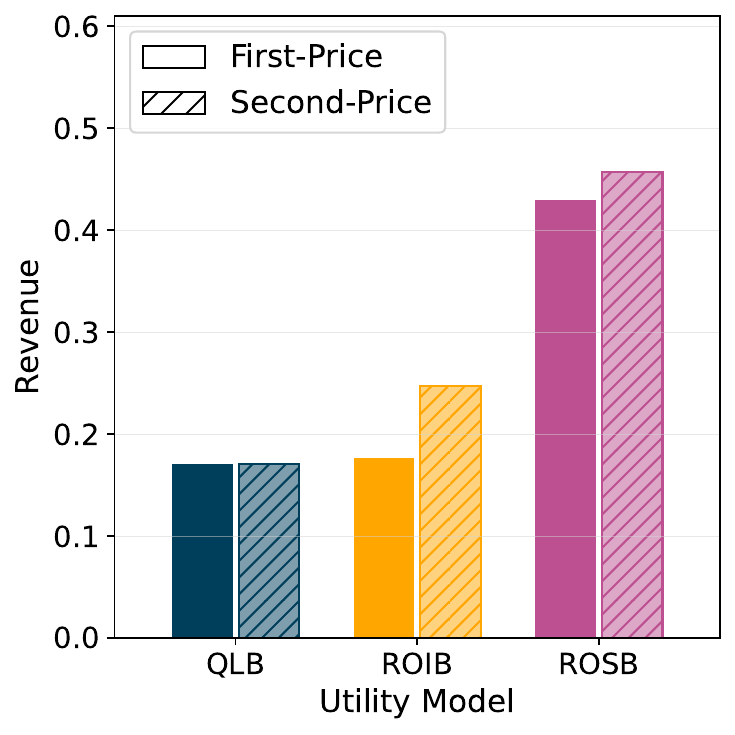}
            \caption{Budget Constraints (B=1.01)}
        \end{subfigure}
        \begin{subfigure}{0.32\textwidth}
            \includegraphics[width=0.9\textwidth]{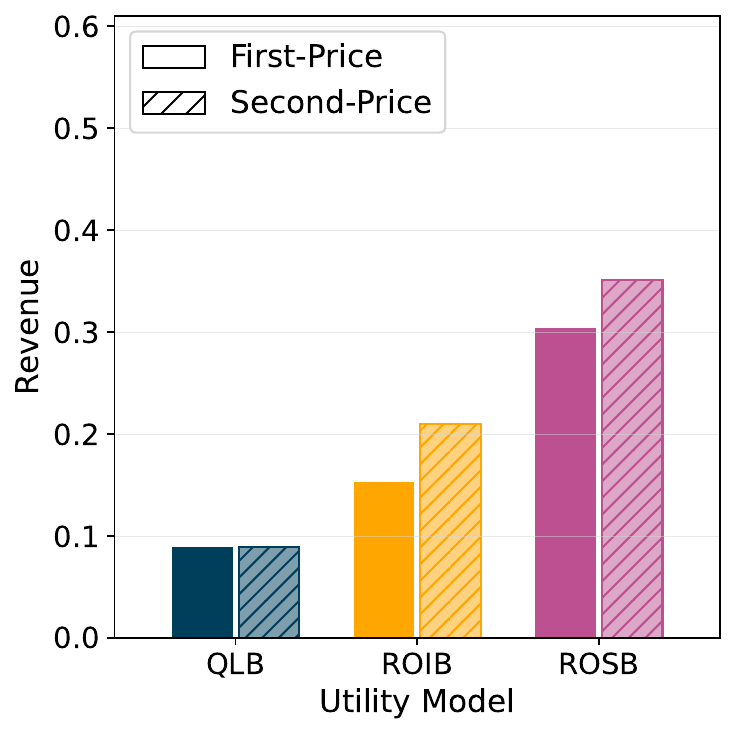}
            \caption{Budget Constraints (B=0.81)}
        \end{subfigure}
        \caption{Revenue for Two Bidders and Different Variations of the Utility Functions.}
        \label{fig:soda_robust}
    \end{center}
    \footnotesize
    We use the computed strategies of SODA to approximate the expected revenue for the first- and second-price auctions with $n=2$ bidders and variations of the utility models. The first plot shows the utility model described in Equation \eqref{eq:rois} for different values of the parameter $\mu$. In the second an third plot, we show the revenue of the setting described in Equation \eqref{eq:budget} for different budgets $B$. The revenue is approximated by simulating $2^{22}$ auctions for each setting. We plot the mean over ten runs. The standard deviation in all settings over these runs is less than $ 2 \cdot 10^{-2} $.
\end{figure}

Second, we model budget constraints in all three settings for better comparison. Similar to the ROSB model, we add a log-barrier function to the QL and ROI utility functions and get
\begin{equation} \label{eq:budget}
    u_i^{\alpha B }(b,v_i) =  u_i^{\alpha}(b,v_i) + \log(B - p_i(b)) \quad \text{with } \alpha \in \{\text{QL }, \text{ROI }, \text{ROS}\}.
\end{equation}
To be well defined, the budget has to be slightly larger than the maximal bid agents can submit. To that end, we restrict the action space to $\Acal = [0, B-0.01] $ for a given budget $B > 0$. In our experiments, we consider $B \in \{ 1.01, 0.81\}$. The corresponding revenues for the different utility models are visualized in Figure \ref{fig:soda_robust} (b) and (c). We note that the bids submitted and, thus, the revenue is lower compared to a situation without budgets. This is due to the log-barrier function, which also penalizes payments that are close to the budget. Nevertheless, we can observe similar differences in revenues between the two payment rules for ROIB and ROSB, while the revenues for the payoff-maximizing agents with budget (QLB) remain comparable.

\subsection{Additional Experiments for Bandit Algorithms} \label{app:bandit_bayesian}
In Section \ref{sec:bandit_bayesian}, we already observed that even in the incomplete-information model, we can learn equilibrium strategies using simple bandit algorithms such as Exp3.
Furthermore, we observed the same outcomes with respect to revenue for the different utility models, as predicted by the equilibrium analysis using SODA.

In this section, we report the results for additional bandit algorithms such as $\varepsilon$-Greedy and Thompson-Sampling and for the Q-Learning algorithm used in \citet{banchio2022artificial}. We report the relative utility loss, which allows us to evaluate if the played actions are close to the Bayes-Nash equilibrium. 
Specifically, after $10^7$ learning iterations, we examine the final 40,000 iterations, recording each action along with its respective valuation to construct a distributional strategy (see Section \ref{sec:soda}). 
Examples of these strategies for ROI-maximizing agents in a first-price auction are visualized in Figure \ref{fig:bandit_roi}.
\begin{figure}[h]
    \begin{center}
        \includegraphics[width=0.9\textwidth]{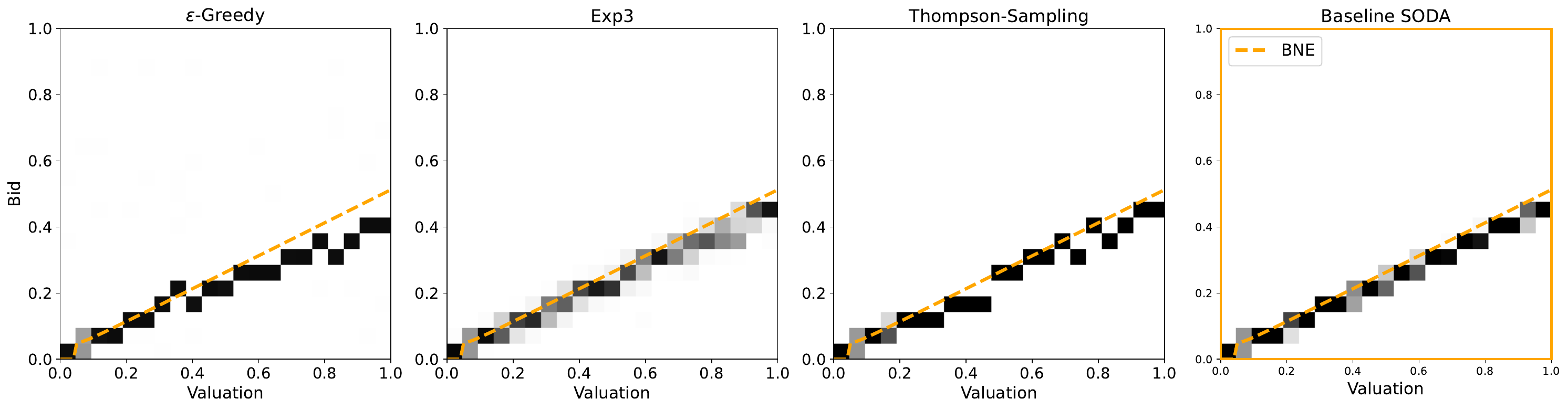}
        \caption{Strategies for ROI-Maximizing Agents using Bandit Algorithms in a FPSB Auction with Uniform Prior.}
        \label{fig:bandit_roi}
    \end{center}
    \footnotesize
    We run $\varepsilon$-Greedy, Exp3, and Thompson Sampling for $10$ million iterations and visualize the frequency of the last $40\thinspace000$ bids for the respective valuations, i.e., the induced distributional strategies. On the last plot, we show the distributional equilibrium strategy computed with SODA. The colored lines denote the BNE in the continuous setting.
\end{figure}

Based on these constructed strategy profiles, we can compute the relative utility loss, representing the relative performance improvement achievable if the agent were to best respond to the opponents' strategies. 
We compare the results with the approximated equilibrium strategies using SODA (with the standard projected gradient update step due to better performance). The setting is identical to the one described in  Section \ref{sec:bandit_bayesian}. 
The results can be found in Table \ref{tab:bandits_bayesian}. 
For a better interpretation of the values: If agents with quasilinear utility function would collude and bid only half of the amount they would bid in equilibrium, i.e., $\tfrac 1 3 v$, the relative utility loss would be at $\ell = 0.20$ with expected revenue of $0.25$. The relatively high numbers for $\varepsilon$-Greedy are due to the constant exploration factor $\varepsilon$. 
Other than that, we observe that bandit algorithms such as Exp3 and Thompson-Sampling perform very well and learn to play an approximate Bayes-Nash equilibrium.

\input{tables/table_bandits_bayesian}

Additionally, we also report the average revenue of the last 40,000 iterations. In all settings, we closely approximate the revenue as predicted by the equilibrium analysis. In particular, we find that when the utility loss is relatively high, e.g., for Q-learning in the first-price auction with ROI-maximizing agents, the return is higher than expected, indicating a failure to learn the optimal strategy rather than a collusive behavior of the agents.

	\clearpage
	\newpage 
\end{appendix}

\end{document}

%% file: tables/table_learner_compl_info.tex
\begin{table}[h]
    \caption{Parameters for Q-Learner and Bandit Algorithm.}
    \begin{center}
        \begin{tabular}{ l r r r r r r }
            \toprule
            Learner & Init. & $\epsilon$ & $\beta$ &  $\gamma$ & $\alpha$ & $Q_0$ \\
            \midrule
            Optimistic (Baseline) & 100 & 0.025 & 0.0002 & 0.99 & 0.05 & 100 \\
            Optimistic \& small $\gamma$ & 100 & 0.025 & 0.0002 & 0.25 & 0.05 & 100 \\
            Optimistic \& fix $\varepsilon$ & 100 & 0.025 & 0.0 & 0.99 & 0.05 & 100 \\
            Zero & 0 & 0.025 & 0.0 & 0.99 & 0.05 & 0 \\
            Exp3 & uniform & 0.01 & 0.0 & - & 0.01 & - \\
            \bottomrule
        \end{tabular}
    \end{center}
    \label{tab:param_compl_info}
    \footnotesize
    Optimistic (Baseline) corresponds to the parameters from \citet{banchio2022artificial}.
\end{table}

%% file: tables/table_soda.tex
\begin{table}[h]
  \caption{Evaluation of SODA against Analytical Equilibrium Strategies.}
  \begin{center}
      \begin{tabular}{ l l c c c c }
          \toprule
          Utility Model & Pricing Rule & Bidders & $\ell$ & $\mathcal L$ & $L_2$ \\
          \midrule
          \multirow{8}{*}{Quasi-Linear (QL)} & \multirow{4}{*}{First-Price} 
            & 2  & 0.000 (0.000) & 0.002 (0.000) & 0.010 (0.000) \\
          & & 3  & 0.000 (0.000) & 0.001 (0.000) & 0.010 (0.000) \\
          & & 5  & 0.000 (0.000) & 0.002 (0.000) & 0.023 (0.001) \\
          & & 10 & 0.001 (0.000) & 0.011 (0.000) & 0.084 (0.002) \\
          \cmidrule(lr){2-6}
          & \multirow{4}{*}{Second-Price} 
            & 2  & 0.000 (0.000) & 0.000 (0.000) & 0.013 (0.000) \\
          & & 3  & 0.000 (0.000) & 0.001 (0.000) & 0.014 (0.000) \\
          & & 5  & 0.000 (0.000) & 0.002 (0.000) & 0.028 (0.001) \\
          & & 10 & 0.000 (0.000) & 0.011 (0.000) & 0.092 (0.001) \\
          \midrule
          \multirow{8}{*}{Return-on-Invest (ROI)} & \multirow{4}{*}{First-Price} 
            & 2  &  0.000 (0.000) & 0.014 (0.001) & 0.013 (0.000) \\
          & & 3  &  0.000 (0.000) & 0.004 (0.000) & 0.011 (0.001) \\
          & & 5  &  0.000 (0.000) & 0.002 (0.000) & 0.014 (0.000) \\
          & & 10 &  0.001 (0.001) & 0.005 (0.000) & 0.070 (0.001) \\
          \cmidrule(lr){2-6}
          & \multirow{4}{*}{Second-Price} 
            & 2  & 0.000 (0.000) &  0.000 (0.000) & 0.016 (0.000) \\
          & & 3  & 0.000 (0.000) &  0.001 (0.000) & 0.013 (0.000) \\
          & & 5  & 0.000 (0.000) &  0.002 (0.000) & 0.020 (0.000) \\
          & & 10 & 0.000 (0.000) &  0.009 (0.000) & 0.077 (0.000) \\
          \bottomrule
      \end{tabular}
  \end{center}
  \label{tab:results_soda}
  \footnotesize
  The mean (and standard deviation) of the relative utility loss in the discretized game $\ell$ and the approximated utility loss $\mathcal L$ and the $L_2$ norm with respect to the analytical BNE are reported for different utility models (payoff-maximizing, ROI-maximizing), payment rules (first- and second-price), and different numbers of agents ($N \in \{2,3, 5, 10\}$). In all settings, we assume a uniform prior.
\end{table}

%% file: tables/table_bandits_bayesian.tex
\begin{table}[h]
  \caption{Evaluation of Bandit Algorithms in the Incomplete-Information Setting.}
  \begin{center}
      \begin{tabular}{ l r r c c }
          \toprule
          Utility Model & Pricing Rule & Learner & Rel. Utility Loss $\ell$ & Revenue\\
          \midrule
          \multirow{8}{*}{QL} & \multirow{4}{*}{First-Price} & SODA  & $<0.001$ & 0.496 \\
          & & $\varepsilon$-Greedy        & 0.129 (0.006) & 0.509 (0.006) \\ 
          & & Exp3                        & 0.025 (0.001) & 0.497 (0.003) \\
          & & Thompson-Sampling           & 0.003 (0.000) & 0.494 (0.001) \\
          & & Q-Learning (Baseline) & 0.021 (0.017) & 0.506 (0.002) \\ \cmidrule(lr){2-5}
          & \multirow{4}{*}{Second-Price} & SODA  & $<0.001$ & 0.497 \\
          & & $\varepsilon$-Greedy        & 0.053 (0.003) & 0.496 (0.004) \\ 
          & & Exp3                        & 0.017 (0.003) & 0.496 (0.002) \\
          & & Thompson-Sampling           & 0.007 (0.000) & 0.497 (0.002) \\
          & & Q-Learning (Baseline) & 0.010 (0.011) & 0.499 (0.003) \\ 
          \midrule
          \multirow{8}{*}{ROI} & \multirow{4}{*}{First-Price} & SODA  & $<0.001$ & 0.344 \\
          & & $\varepsilon$-Greedy        & 0.113 (0.038) & 0.355 (0.016) \\ 
          & & Exp3                        & 0.026 (0.006) & 0.350 (0.002) \\
          & & Thompson-Sampling           & 0.068 (0.027) & 0.349 (0.009) \\
          & & Q-Learning (Baseline) & 0.140 (0.051) & 0.406 (0.004) \\ \cmidrule(lr){2-5}
          & \multirow{4}{*}{Second-Price} & SODA  & $<0.001$ & 0.502 \\
          & & $\varepsilon$-Greedy        & 0.031 (0.005) & 0.499 (0.006)\\ 
          & & Exp3                        & 0.031 (0.006) & 0.486 (0.008) \\
          & & Thompson-Sampling           & 0.016 (0.008) & 0.489 (0.009) \\
          & & Q-Learning (Baseline) & 0.010 (0.004) & 0.492 (0.006) \\ 
          \midrule
          \multirow{8}{*}{ROSB} & \multirow{4}{*}{First-Price} & SODA  & $<0.001$ & 0.544 \\
          & & $\varepsilon$-Greedy        & 0.241 (0.074) & 0.523 (0.039) \\ 
          & & Exp3                        & 0.053 (0.004) & 0.538 (0.001) \\
          & & Thompson-Sampling           & 0.002 (0.000) & 0.543 (0.000) \\
          & & Q-Learning (Baseline) & 0.033 (0.029) & 0.545 (0.006) \\ \cmidrule(lr){2-5}
          & \multirow{4}{*}{Second-Price} & SODA  & $<0.001$ & 0.555 \\
          & & $\varepsilon$-Greedy        & 0.056 (0.004) & 0.551 (0.003) \\ 
          & & Exp3                        & 0.026 (0.008) & 0.551 (0.003) \\
          & & Thompson-Sampling           & 0.011 (0.006) & 0.554 (0.002) \\
          & & Q-Learning (Baseline) & 0.012 (0.012) & 0.551 (0.003) \\
          \bottomrule
      \end{tabular}
  \end{center}
  \label{tab:bandits_bayesian}
  \footnotesize
  The mean (and standard deviation) of the relative utility loss $\ell$ and the average revenue for one agent in the discretized game are reported for auction formats with $N=3$ agents, uniform prior, different utility models (QL, ROI, and ROSB), and different payment rules (first- and second-price). The mean and standard deviation is computed over 10 runs for each experiment.
\end{table}